\documentclass[11]{llncs}
\usepackage{times}
\usepackage{compress}
\usepackage{epsfig}
\usepackage{subfigure}
\usepackage{gensymb}
\usepackage{amssymb}
\usepackage{amsmath}
\usepackage{graphicx}
\usepackage{wrapfig}
\usepackage{color}
\usepackage{textcomp}
\usepackage{tabularx}
\usepackage{latexsym}
\usepackage{caption}
\usepackage{paralist}
\usepackage{todonotes}

\graphicspath{{figures2/}} 

\newcommand{\remove}[1]{}

\def\comment#1{}

\def\ITEMMACRO #1 ??? #2 ???{\par\vskip2pt\noindent%
\hangindent=#2em\setbox0\hbox{#1\kern4pt}%
\ifdim\wd0<\hangindent\setbox0\hbox to\hangindent{\hss#1\kern.66em}\fi%
\box0\ignorespaces}

\begin{document}

\title{A Universal Point Set for 2-Outerplanar Graphs
}

\author{
  Patrizio~Angelini\inst{1}
  \and
  Till~Bruckdorfer\inst{1}
  \and 
  Michael~Kaufmann\inst{1}
  \and
  Tamara~Mchedlidze\inst{2}
}

\institute{
  Wilhelm-Schickard-Institut f\"ur Informatik, Universit\"at
  T\"ubingen, Germany
  \and
  Institute of Theoretical Informatics, Karlsruhe Institute of Technology,
  Germany
}

\maketitle

\begin{abstract}
A point set $S \subseteq \mathbb{R}^2$ is universal for a class $\cal G$ 
if every graph of ${\cal G}$ has a planar straight-line
embedding on $S$. 
It is well-known that the integer grid is a quadratic-size universal point set 
for planar graphs, while the existence of a sub-quadratic universal point set for them
is one of the most fascinating open problems in Graph Drawing.  
Motivated by the fact that outerplanarity is a key property for the existence of small universal point sets, 
we study $2$-outerplanar graphs and provide for them a universal point set of size 
$O(n \log n)$.
\end{abstract}

\section{Introduction}
\label{ups:introduction}

Let $S$ be a set of $m$ points on the plane. A \emph{planar straight-line 
	embedding} of 
an $n$-vertex planar graph $G$, with $n\leq m$, on $S$ is a mapping of 
each vertex of $G$ to a distinct point of $S$ so that, if the edges are drawn straight-line, no two edges cross. Point set $S$ is
\emph{universal} for a class $\cal G$ of graphs if every graph $G \in {\cal G}$ has a planar straight-line embedding on $S$. Asymptotically, the smallest universal point set for general planar graphs
is known to have size at least $1.235n$~\cite{Kurowski04},
while the upper bound is $O(n^2)$~\cite{DBLP:journals/jgaa/BannisterCDE14,FraysseixPP90,Schnyder90}. 
All the upper bounds are based on drawing the graphs on an integer grid, except for the one by Bannister et al.~\cite{DBLP:journals/jgaa/BannisterCDE14}, who use super-patterns to obtain a universal point set of size $n^2/4-\Theta(n)$ -- currently the best result for planar graphs. Closing the gap between the lower and the upper bounds is a challenging open problem~\cite{Cabello06,DBLP:conf/stoc/FraysseixPP88,FraysseixPP90}. 

A subclass of planar graphs for which the ``smallest possible'' universal point set
is known is the class of \emph{outerplanar} graphs -- the graphs that
admit a straight-line planar drawing in which all vertices are incident to the
outer face. Namely, Gritzmann et al.~\cite{GritzmannMPP91} and
Bose~\cite{comgeo/Bose02} proved that any point set of size $n$ in general position is
universal for $n$-vertex outerplanar graphs. 
Motivated by this result, we consider the class of \emph{$k$-outerplanar} 
graphs, with $k \geq 2$, 
which is a generalization of outerplanar graphs. A planar drawing of a graph is 
$k$-outerplanar 
if removing the vertices of the outer face, called $k$-th \emph{level}, produces a 
$(k-1)$-outerplanar drawing, where $1$-outerplanar stands for outerplanar. 
A graph is $k$-outerplanar if it admits a $k$-outerplanar drawing.
Note that every planar graph is a $k$-outerplanar graph, for some value of $k 
\in O(n)$. 
Hence, in order to tackle a meaningful subproblem of the general one, it makes 
sense to 
study the existence of subquadratic universal point sets when the value of $k$ 
is bounded 
by a constant or by a sublinear function. However, while the case $k=1$ is 
trivially solved 
by selecting any $n$ points in general position, as observed 
above~\cite{comgeo/Bose02,GritzmannMPP91}, 
the case $k=2$ already eluded several attempts of solution and turned out to be far from 
trivial. 
In this paper, we finally solve the case $k=2$ by providing a universal point 
set for $2$-outerplanar graphs of size $O(n \log n)$. 

A subclass of $k$-outerplanar graphs, in which  the value 
of $k$ is unbounded, but every level is restricted to be a chordless simple 
cycle, was known to have a universal point set of size $O(n (\frac{\log n}{\log\log 
	n})^2)$~\cite{AngeliniBKMRS11}, which was subsequently reduced to $O(n 
\log n)$~\cite{DBLP:journals/jgaa/BannisterCDE14}. 
It is also known that \emph{planar 3-trees} -- graphs not defined in terms of 
$k$-outerplanarity -- have a universal point set of size $O(n^{5/3})$~\cite{DBLP:conf/wads/FulekT13}. Note that planar $3$-trees have treewidth equal to $3$, while $2$-outerplanar graphs have treewidth at most $5$.

\noindent 
{\bf Structure of the paper:} 
After some preliminaries and definitions in Section~\ref{section:preliminaries},  
we consider $2$-outerplanar graphs in Section~\ref{section:triangulatedForests} where the inner level 
is a forest and all the internal faces are triangles. We prove that this class of graphs admits a 
universal point set of size $O(n^{3/2})$. 
We then extend the result in Section~\ref{section:Forests} to $2$-outerplanar graphs in which the inner 
level is still a forest but the faces are allowed to have larger size.
Finally, in Section~\ref{section:contraction}, we outline how the result of Section~\ref{section:Forests} can 
be extended to general $2$-outerplanar graphs. 
We also explain how to apply the methods by Bannister et al. in~\cite{DBLP:journals/jgaa/BannisterCDE14} 
to reduce the size of the point set to $O(n \log n)$.
We conclude with open problems in Section~\ref{section:conclusions}.

\section{Preliminaries and Definitions}
\label{section:preliminaries}

In this section we introduce basic terminology used throughout the paper. 
A straight-line segment with endpoints 
$p$ and $q$ is denoted by $s(pq)$. A circular arc with endpoints $p$ and $q$ 
(clockwise) is denoted by $a(pq)$. We assume familiarity with the concepts of \emph{planar graphs}, 
\emph{straight-line planar drawings}, and their \emph{faces}. 
A straight-line planar drawing $\Gamma$ of a graph $G$ determines a clockwise ordering 
of the edges incident to each vertex $u$ of $G$, called \emph{rotation at} 
$u$. The \emph{rotation scheme} of $G$ in $\Gamma$ is the set of the rotations 
at all the vertices of $G$ determined by $\Gamma$. Observe that, if $G$ is 
connected, in all the straight-line planar drawings of $G$ determining the same 
rotation scheme, the faces of the drawing are delimited by the same edges.

Let $[G,{\cal H}]$ be a $2$-outerplanar graph, where the outer level 
is an outerplanar graph $G$ and the inner level is a set ${\cal H} = 
\{G_1,\dots,G_k\}$ of outerplanar graphs. We assume that $[G,{\cal H}]$ is given 
together with a rotation scheme, and the goal is to construct a planar 
straight-line embedding of $[G,{\cal H}]$ on a point set determining this 
rotation scheme. Since $[G,{\cal H}]$ can be assumed to be connected (as 
otherwise we can add a minimal set of dummy edges to make it connected), this is 
equivalent to assuming that a straight-line planar drawing $\Gamma$ of $[G,{\cal 
	H}]$ is given. We rename the faces of $\Gamma$ as $F_1,\dots,F_k$ in such a 
way that each graph $G_h$, which can also be assumed connected, lies inside face 
$F_h$. Note that, for each face $F_h$ of $G$, the graph $[F_h,G_h]$ is again a 
$2$-outerplanar graph; however, in contrast to $[G,{\cal H}]$, its outer level 
$F_h$ is a simple chordless cycle and its inner level $G_h$ consists of only one 
connected component. In the special case in which $G_h$ is a tree we say that 
graph $[F_h,G_h]$ is a \emph{cycle-tree} graph.  We say that a $2$-outerplanar graph is \emph{inner-triangulated} if all the internal faces are $3$-cycles. Note that not every cycle-tree graph can be augmented to be inner-triangulated without introducing multiple edges.

\section{Inner-Triangulated $2$-Outerplanar Graphs with Forest}
\label{section:triangulatedForests}

In this section we prove that there exists a universal point set $S$ of size $O(n^{3/2})$ for the class of $n$-vertex inner-triangulated $2$-outerplanar graphs $[G,{\cal H}]$ where ${\cal H}$ is a forest. 
%

\subsection{Construction of the Universal Point Set}
\label{subsection:Construction}

In the following we describe $S$; refer to Fig.~\ref{fig:universalpointset}.
Let $\pi$ be a half circle with center $O$ and let $N:=n+\sqrt{n}$. 
Uniformly distribute points in $S_{\cal M}=\{p_1,\dots,p_N\}$ on $\pi$. The points in
$S_{\cal D}=\{p_{i\sqrt{n}+i}: 1\leq i \leq \sqrt{n} \} $ are called \emph{dense}, 
while the remaining points in $S_{\cal M} \setminus S_{\cal D}$ are  \emph{sparse}\footnote{The distribution of the points into dense and sparse portions of the point set is inspired by~\cite{AngeliniBKMRS11}.}.

\begin{figure}
	\centering
	\includegraphics[width=0.55\textwidth]{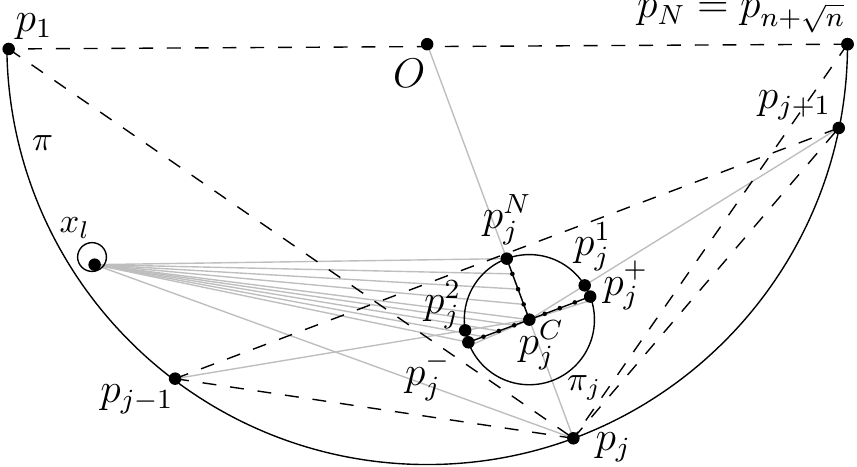}
	\caption{Illustration of $S$, focused on $S_j$ of $p_j$.}
	\label{fig:universalpointset}
\end{figure}

For $j = 2, \dots, N-1$, place a circle 
$\pi_j$ with its center $p_j^C$ on $s(p_jO)$, so that it lies completely inside 
the triangle $\triangle p_{j-1}p_jp_{j+1}$ and inside the triangle $\triangle 
p_1p_jp_N$. Note that the angles $\angle p_jp_j^Cp_N$ and $\angle 
p_jp_j^Cp_1$ are smaller than $180^\circ$. Let $p_j^N$ be the intersection point 
between $s(p_jO)$ and $\pi_j$ that is closer to $O$. 
Also, let $p_j^1$ (resp. $p_j^2$) be the intersection point of $s(p_j^Cp_{j+1})$ 
(resp. $s(p_j^Cp_{j-1})$) with $\pi_j$. Finally, let $p_j^3$ (resp. $p_j^4$) 
be the intersection point of $\pi_j$ with its diameter orthogonal to 
$s(p_jO)$, such that $a(p_j^3p_j^4)$ does not contain $p_j^N$.
Now, choose a point $p_j^+$ on the arc $a(p_j^1p_j^3)$, and a point $p_j^-$ on 
the arc $a(p_j^4p_j^2)$. To complete the construction of $S$, evenly distribute 
$\overline{n}-1$ points on each of the three segments $s^N_j := s(p_j^Cp_j^N)$, 
$s^+_j := s(p_j^Cp_j^+)$, and $s^-_j := s(p_j^Cp_j^-)$, where $\overline{n} = 
n$ if $p_j$ is dense and $\overline{n} = \sqrt{n}$ if it is sparse. We refer to 
the points on $s^N, s^+, s^-$, including the points $p_j^N, p_j^C, p_j^+, 
p_j^-$, as \emph{the point set of} $p_j$, and we denote it by $S_j$. Vertex 
$p_j^C$ is the \emph{center vertex of} $S_j$.

The described construction uses $(\sqrt{n}-1)(3n+1) + (n-1)(3\sqrt{n}+1)$=$O(n^{3/2})$ 
points and ensures the following property.

\begin{property}\label{property:visibility} $ $
	For each $j = 1, \dots, N$, the following visibility properties hold:\\
	\begin{enumerate}[(A)]
		\item 
		The straight-line segments connecting point $p_j$ to: point $p_j^-$, to the 
		points on $s^-_j$, to $p_j^C$, to the points on $s^+_j$, and to $p_j^+$ appear 
		in this clockwise order around $p_j$.
		\item For all $l < j$, consider any point $x_l \in \{p_l\} \cup S_l$ (see 
		Fig.~\ref{fig:universalpointset}); then, the straight-line segments connecting 
		$x_l$ to: $p_j^N$, to the points on $s^N_j$, to $p_j^C$, to the points on 
		$s^-_j$, to $p_j^-$, and to $p_j$ appear in this clockwise order around $x_l$. 
		Also, consider the line passing through $x_l$ and any point in $\{p_j\} \cup 
		S_j$; then, every point in $\{p_q\} \cup S_q$, with $l < q < j$, lies in the 
		half-plane delimited by this line that does not contain the center point $O$ of 
		$\pi$.
		\item For all $l > j$, consider any point $x_l \in \{p_l\} \cup S_l$; then, the 
		straight-line segments connecting $x_l$ to: $p_j^N$, to the points on $s^N_j$, to 
		$p_j^C$, to the points on $s^+_j$, to $p_j^+$, and to $p_j$ appear in this 
		counterclockwise order around $x_l$. Also, consider the line passing through 
		$x_l$ and any point in $\{p_j\} \cup S_j$; then, every point in $\{p_q\} \cup 
		S_q$, with $j < q < l$, lies in the half-plane delimited by this line that does 
		not contain $O$.
	\end{enumerate}
\end{property}
\begin{proof}
	Item (A) follows from the fact that $p_j^-$ and $p_j^+$ lie on different sides 
	of segment $s(p_jO)$. In order to prove item (B), consider the intersection 
	point $p_x$ between $\pi_j$ and segment $s(p_j^Cx_l)$; then, the first statement 
	of item (B) follows from the fact that points $p_j^-$, $p_x$, and $p_j^N$ appear 
	in this clockwise order along $\pi_j$. This is true since, by the construction 
	of $S$, point $p_x$ lies between $p_j^2$ and $p_j^N$, and point $p_j^-$ precedes 
	$p_j^2$ in this clockwise order.
	As for the second statement, this depends on the fact that each point set $S_q$, 
	with $l < q < j$, is entirely contained inside triangle $\triangle 
	p_{q-1},p_q,p_{q+1}$. The proof for item (C) is symmetrical to the one for item 
	(B).
	\qed
\end{proof}	

\subsection{Labeling the Graph}
\label{subsection:Labeling}

Let $[G,{\cal H}]$ be an inner-triangulated $2$-outerplanar graph where $G$ is 
an outerplanar graph and ${\cal H}=\{T_1,\dots,T_k\}$ is a forest such that tree 
$T_h$ lies inside face $F_h$ of $G$, for each $1 \leq h \leq k$. 
The idea behind the labeling is the following: in our embedding strategy, $G$ 
will be embedded on the half-circle $\pi$ of the point set $S$, while the 
tree $T_h \in {\cal H}$ lying inside each face $F_h$ of $G$ will be embedded on 
the point sets $S_j$ of some of the points $p_j$ on which vertices of $F_h$ are 
placed. Note that, since $\pi$ is a half-circle, the drawing of $F_h$ will 
always be a convex polygon in which two vertices have \emph{small} (acute) 
internal angles, while all the other vertices have \emph{large} (obtuse) 
internal angles. In particular, the vertices with the small angle are the first 
and the last vertices of $F_h$ in the order in which they appear along the outer face of 
$\Gamma$. Since, by construction, a point $p_j$ of $F_h$ has its point set $S_j$ 
in the interior of $F_h$ if and only if it has a large angle, we aim at 
assigning each vertex of $T_h$ to a vertex of $F_h$ that is neither the first 
nor the last. We will describe this 
assignment by means of a labeling $\ell \colon [G,{\cal H}] \rightarrow 1, \dots, |G|$; namely, 
we will assign a distinct label $\ell(v)$ to each vertex $v \in G$ and then 
assign to each vertex of $T_h$ the same label as one of the vertices of $F_h$ 
that is neither the first or the last. Then, the number of vertices with the same label as a vertex of $G$ will determine whether this vertex 
will be placed on a sparse or a dense point. We formalize this idea in the following.

We rename the vertices of $G$ as $v_1, \dots, 
v_{|G|}$ in the order in which they appear along the outer face of $\Gamma$, and label 
them with $\ell(v_i)=i$ for $i=1,\dots,|G|$. Next, we label the vertices 
of each tree $T_h \in {\cal H}$. Since trees $T_h$ and $T_{h'}$ are disjoint for $h \neq 
h'$, we focus on the cycle-tree graph $[F,T]$ composed of a single face $F=F_h$ of $G$ and 
of the tree $T=T_h \in {\cal H}$ inside it.
Rename the vertices of $F$ as $w_1, \dots, w_m$ in such a way that for any two 
vertices $w_x=v_p$ and $w_{x+1}=v_q$, where $p,q \in \{1,\dots,|G|\}$, it holds 
that $p < q$. As a result, $w_1$ and $w_m$ 
are the only vertices of $F$ with small internal 
angles. A vertex of $T$ is a \emph{fork vertex} if it is adjacent to 
more than two vertices of $F$ (square vertices in 
Fig.~\ref{fig:newlabelingfigure-a}), otherwise it is a \emph{non-fork vertex} 
(cross vertices in Fig.~\ref{fig:newlabelingfigure-a}). Since $[F,T]$ is 
inner-triangulated, every vertex of $T$ is adjacent to at least two vertices of 
$F$, and hence non-fork vertices are adjacent to exactly two vertices of $F$. 

\begin{figure}[tbh]
	\centering
	\subfigure[\label{fig:newlabelingfigure-a}]{\includegraphics[
		width=0.42\textwidth]{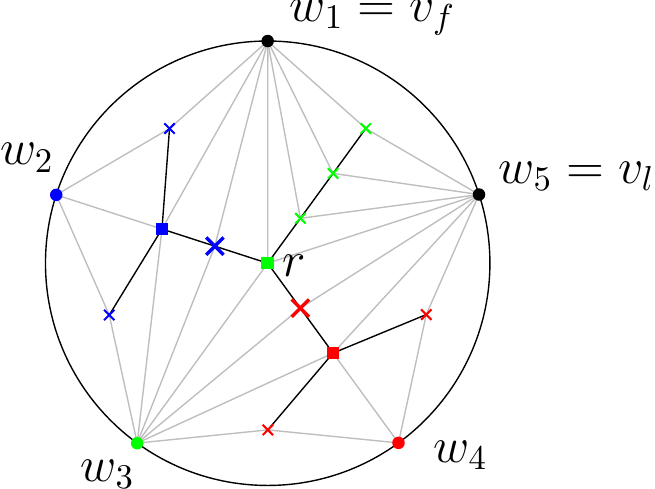}}
	\hfil
	\subfigure[\label{fig:polygon-b}]{\includegraphics[width=0.49\textwidth]{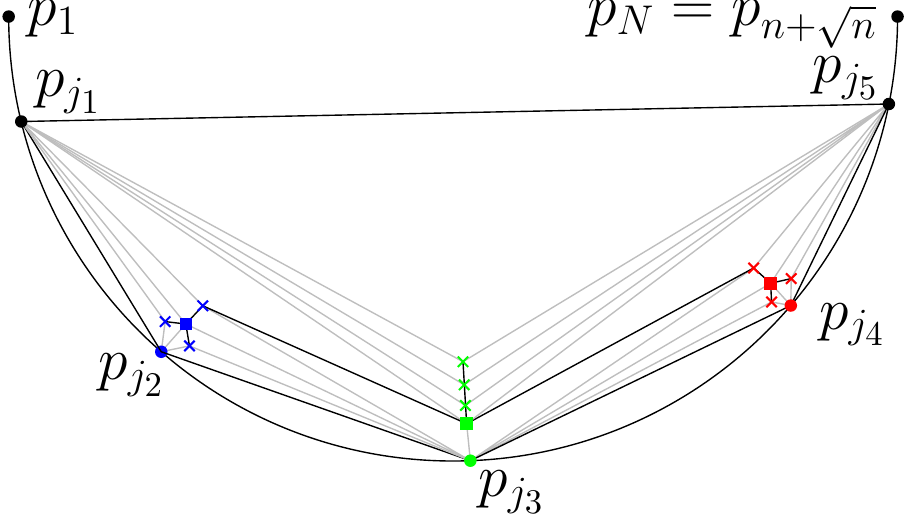}}
	\caption{\subref{fig:newlabelingfigure-a} A cycle-tree graph $[F,T]$ 
		with $F=\{w_1,w_2,w_3,w_4,w_5\}$, where $\ell(w_2)$ is blue, 
		$\ell(w_3)$ is green and $\ell(w_4)$ is red. Fork vertices are squares; foliage vertices 
		are small crosses, while branch vertices are large crosses. Tree $T'$ is 
		composed  of the root $r$ (the green square vertex) with two children (the 
		red and the  blue square vertices). Vertices of $T$ got color red, green, 
		blue according to the labeling algorithm.
		\subref{fig:polygon-b} An embedding of $[F,T]$ according to Steps a,b, and c.}
	\label{fig:labeling}
\end{figure}
We label the vertices of $T$ starting from its fork vertices. To this end, we construct a tree $T'$ 
composed only of the fork vertices, as 
follows. Initialize $T'$=$T$. Then, as long as there exists a non-fork vertex of 
degree $3$ (namely, with $2$ neighbors in $F$ and $1$ in $T'$), remove it and 
its incident edges from $T'$. The vertices removed in this step are called \emph{foliage} 
(small crosses in 
Fig.~\ref{fig:newlabelingfigure-a}). All the remaining non-fork vertices have 
degree $4$ (namely $2$ in $F$ and $2$ in $T'$); for each of them, remove it and 
its incident edges from $T'$ and add an edge between the two vertices of $T'$ 
that were connected to it before its removal. The vertices removed in this step are 
\emph{branch} vertices (large crosses in 
Fig.~\ref{fig:newlabelingfigure-a}).
A vertex $w_x \in F$ is called \emph{free} if so far 
no vertex of $T'$ has label $\ell(w_x)$. To perform the labeling, we 
traverse $T'$ bottom-up with respect to a root $r$ that is the vertex of $T'$ 
adjacent to both $w_1$ and $w_m$. Since $[F,T]$ is inner-triangulated, this 
vertex is unique. During the traversal of $T'$, we maintain the invariant that 
vertices of $T'$ are incident to only free vertices of $F$.
Initially the invariant is satisfied since all the vertices 
of $F$ are free. Let $a$ be the fork vertex considered in a step of the 
traversal of $T'$, and let $w_{a_1},\dots,w_{a_k}$ be the vertices of $F$ 
adjacent to $a$, with $1 \leq a_1 <\dots<a_k \leq m$ and $k \geq 3$. By the 
invariant, $w_{a_1},\dots,w_{a_k}$ are free. Choose any vertex $w_{a_i}$ 
such that $2 \leq i \leq k-1$, and set $\ell(a)=\ell(w_{a_i})$. 
For example, the red fork vertex in Fig.~\ref{fig:newlabelingfigure-a} adjacent to 
$w_3$, $w_4$, and $w_5$ in $F$ gets label $\ell(w_4)$. 
Since vertices $w_{a_2},\dots,w_{a_{k-1}}$ 
cannot be adjacent to any vertex of $T'$ that is visited after $a$ in the 
bottom-up traversal, the invariant is maintained at the end of each step.
At the last step of the traversal, 
when $a=r$, we have that $w_{a_1}=w_1$ and $w_{a_k}=w_m$, which are both free.


Now we label the non-fork vertices of $T$ based on the labeling of $T'$. Let $b$ 
be a non-fork vertex. If $b$ is a branch vertex, then consider the first 
fork vertex $a$ encountered on a path from $b$ to a leaf of $T$; set 
$\ell(b)=\ell(a)$. Otherwise, $b$ is a foliage vertex. In this case, consider 
the first fork vertex $a'$ encountered on a path from $b$ to the root $r$ of 
$T$. Let $v,w \in F$ be the two vertices of $F$ adjacent to $b$; assume 
$\ell(v) < \ell(w)$. If $\ell(a') \leq \ell(v)$, then set $\ell(b)=\ell(v)$; if 
$\ell(a') \geq \ell(w)$, then set $\ell(b)=\ell(w)$; and if $\ell(v) < \ell(a') 
< \ell(w)$, then set $\ell(b)=\ell(a')$ (the latter case only happens
when $a'$ is the root and $b$ is adjacent to $w_1$ and $w_m$). 
Note that the described algorithm ensures that adjacent non-fork vertices have the same label.
%
We perform the labeling procedure for every $T_h \in {\cal H}$ and obtain a 
labeling for $[G,{\cal H}]$. For each $i= 1, \dots, |G|$, we say that the subgraph of ${\cal H}$ induced 
by all the vertices of ${\cal H}$ with label $i$ is the \emph{restricted subgraph} 
$H_i$ of ${\cal H}$ for $i$ (see Fig.~\ref{fig:newlabelingfigure-a}). 

\begin{lemma}\label{lem:inducedsubtree}
	The restricted subgraph $H_i$ of $\cal H$, for each $i = 1, \dots, |G|$, 
	is a tree all of whose vertices have degree at most $2$, except for one 
	vertex that may have degree $3$. 
\end{lemma}

\begin{proof}
	First observe that, due to the procedure used to label the vertices of $T'$, 
	graph $H_i$ contains at most one fork vertex $a$, which is hence the only one 
	that may have degree larger than $2$. Since adjacent non-fork vertices got the same 
	label, $H_i$ is connected and only contains paths of non-fork vertices incident 
	to $a$. We prove that there exist at most three of such paths. First, $H_i$ 
	contains at most one path of branch vertices incident to $a$, namely the one 
	connecting it to its unique parent in $T'$. Further, $H_i$ contains at most two 
	paths of foliage vertices incident to $a$, namely one composed of the 
	foliage vertices adjacent to $w_x$ and to $w_{x-1}$, and one composed of the 
	foliage vertices adjacent to $w_x$ and to $w_{x+1}$, where $w_{x-1}, w_x, 
	w_{x+1} \in G$ and $\ell(w_x)=i$. Note that, if $a$ coincides with the root $r$ 
	of $T$, there might exist three paths of foliage vertices incident to $a$, 
	namely the two that are incident to $w_x$, $w_{x-1}$, and $w_{x+1}$, as before, 
	plus one composed of the foliage vertices that are incident to both $w_1$ and 
	$w_m$; however, since $r$ has no parent in $T'$, there is no path of 
	branch vertices incident to $a$ in this case. This concludes the proof of the 
	lemma.
	\qed
\end{proof}

\subsection{Embedding on the Point Set} 
\label{subsection:Embedding}

We describe an embedding algorithm consisting of three steps (see Fig.~\ref{fig:polygon-b}). 

\noindent {\bf Step a:} 
Let $\omega: G \rightarrow \mathbb{N}$ be a weight function with $\omega(v_i) = |\{v \in [G,{\cal 
	H}] \mid \ell(v) = i\}|$ for every $v_i \in G$. Note that $\sum_{v_i \in 
	G}{\omega(v_i)} = n$. We categorize each vertex $v_i \in G$ as \emph{sparse} if $1\leq \omega(v_i) \leq \sqrt{n}$, and \emph{dense}
if $\omega(v_i) > 
\sqrt{n}$. Note that there are at most $\sqrt{n}$ dense vertices. 

\noindent {\bf Step b:} 
We draw the vertices $v_1,\dots,v_{|G|}$ of $G$ on the $N := n + 
\sqrt{n}$ points of $\pi$ in the same order as they appear along the outer face 
of $\Gamma$, in such a way that dense (resp. sparse) vertices are placed on 
dense (resp. sparse) points. The resulting embedding $\widetilde{\Gamma}$ 
of $G$ is planar since $\Gamma$ is planar. The construction of $\widetilde{\Gamma}$ implies the following.

\begin{property} \label{property:polygoncontent}
	Let $Q=\{p_{j_1}, \dots, p_{j_m}\} \subseteq \pi$, $j_i < j_{i+1}$, be the polygon representing a face of $G$. Polygon $Q$ contains in its interior all the point sets $S_{j_2}, \dots, S_{j_{m-1}}$.
\end{property}

\noindent {\bf Step c:} 
Finally, we consider forest ${\cal H} = \{T_1,\dots,T_k\}$. We describe the embedding algorithm for a single cycle-tree graph 
$[F,T]$, where $F = w_1, \dots, w_m$ is a face of $G$ and $T \in {\cal H}$ is 
the tree lying inside $F$. We show how to embed the restricted 
subgraph $H_i$, for each vertex $w_x$ of $F$ with label $\ell(w_x)=i$, on the 
point set $S_j$ of the point $p_j$ where $w_x$ is placed. We remark that the 
labeling procedure ensures that $|H_i|+1 = \omega(w_x) \leq |S_j|$; also, by 
Property~\ref{property:polygoncontent}, point set $S_j$ lies inside the polygon 
representing $F$, except for the two points where vertices $w_1$ and $w_m$ have 
been placed.

By Lemma~\ref{lem:inducedsubtree}, $H_i$ has at most one (fork) vertex $a$ of degree 
$3$, while all other vertices have smaller degree. We place $a$, if any, on the 
center point $p_j^C$ of $p_j$. The at most three paths of non-fork vertices 
are placed on segments $s^+_j, s^-_j, s^N_j$ starting from $p_j^C$; 
namely, the unique path of branch vertices is placed on 
$s^N_j$, while the two paths of foliage vertices are placed on $s^+_j$ or 
$s^-_j$ based on whether the vertex of $G$ different from $w_x$ they are 
incident to is $w_{x+1}$ or $w_{x-1}$, respectively. If $a=r$, then the path of 
foliage vertices incident to $w_1$ and $w_m$ is placed on $s^N_j$. 

\begin{figure}[tbh]
	\centering
	\subfigure[\label{fig:embedding-path-a}]{\includegraphics[
		width=0.32\textwidth]{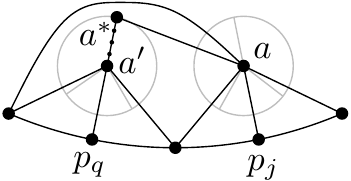}}
	\hfil
	\subfigure[\label{fig:embedding-path-b}]{\includegraphics[width=0.32\textwidth]{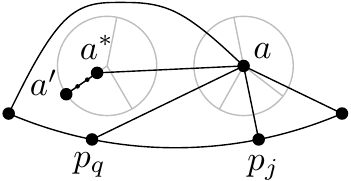}}
	\caption{\subref{fig:embedding-path-a} $P$ contains $a' \neq a$. \subref{fig:embedding-path-b} $a'$ is a leaf of $T$.}
	\label{fig:embedding-path}
\end{figure}

We show that this results in a planar drawing of $T$.
First, for every two fork vertices $a \in H_p$ and $a' \in H_q$, 
with $p < q$, all the leaves of the subtree of $T$ rooted at $a$ have smaller 
label than all the leaves of the subtree of $T$ rooted at $a'$. 
Then, for each $w_x \in F$, with $\ell(w_x)=i$, consider the fork vertex $a \in 
H_i$, which lies on $p_j^C$. Let $P$ be any path 
connecting $a$ to a leaf of $T$ and let $a^*$ be the neighbor of $a$ in $P$. 
If $P$ contains a fork vertex other than $a$ (Fig.~\ref{fig:embedding-path-a}), 
then let $a'$ be the fork vertex 
in $P$ that is closest to $a$ (possibly $a'$=$a^*$) and let $p_q^C$ be the point 
where $a'$ has been placed. Assume $q < j$, the case $q > j$ is analogous. 
By definition, the non-fork vertices in the path from $a$ to $a'$ 
(if any) are branch vertices, and hence lie on 
$s_q^N$. Then, Property~\ref{property:visibility} 
ensures that the straight-line edge $(a,a^*)$ separates all the 
point sets $S_p$ with $q < p < j$ from the center of $\pi$. Since 
the vertices on $S_p$ are only connected either to each other or 
to the vertices on $s_j^-$ and $s_q^+$, edge $(a,a^*)$ is not involved in any crossing.

If $P$ does not contain any fork vertex other than $a$ (Fig.~\ref{fig:embedding-path-b}), 
then all the vertices of 
$P$ other than $a$ are foliage vertices and are placed on a segment $s_q^+$ or 
$s_q^-$, for some $q$. In particular, if $q < j$, then they are on $s_q^-$; if 
$q > j$, then they are on $s_q^+$; while if $q = j$, then they are either on 
$s_q^+$ or on $s_q^-$. In all the cases, 
Property~\ref{property:visibility} ensures that edge $(a,a^*)$ does not cross any edge.

Finally, observe that any path of $T$ containing only non-fork vertices is 
placed on the same segment of the point set, and hence its edges do not cross. As for the edges 
connecting vertices in one of these paths to the two leaves of $T$ they are 
connected to, note that by item $(A)$ of Property~\ref{property:visibility} the 
edges between each of these leaves and these vertices appear in the rotation at 
the leaf in the same order as they appear in the path.


\begin{lemma} \label{theorem:2outerplanartriangulated}
	There exists a universal point set of size $O(n^{3/2})$ for the class of 
	$n$-vertex inner-triangulated $2$-outerplanar graphs $[G,{\cal H}]$ where 
	${\cal 
		H}$ is a forest.
\end{lemma}

\section{$2$-Outerplanar Graphs with Forest}
\label{section:Forests}

In this section we consider $2$-outerplanar graphs $[G,{\cal H}]$ where ${\cal 
	H}$ is a forest. Contrary to the previous section, we do not assume $[G,{\cal H}]$ 
to be inner-triangulated. As observed before, augmenting it might be not possible 
without introducing multiple edges.
The main idea to overcome this problem is to first identify the parts of $[G,{\cal H}]$ 
not allowing for the augmentation, remove them, and augment 
the resulting graph with dummy edges to inner-triangulated 
(Section~\ref{subsection:replacement}); then, apply 
Lemma~\ref{theorem:2outerplanartriangulated} to embed the inner-triangulated 
graph on the point set $S$; and finally remove the dummy edges and embed the 
parts of the graph that had been previously removed on the remaining points 
(Section~\ref{subsection:reverting}). To do so, we first 
need to extend the point set $S$ with some additional points.

\subsection{Extending the Universal Point Set}
\label{subsection:extension}

\begin{figure}[t]
	\subfigure[\label{fig:universalpointset2-a}]{\includegraphics[width=0.49\textwidth]{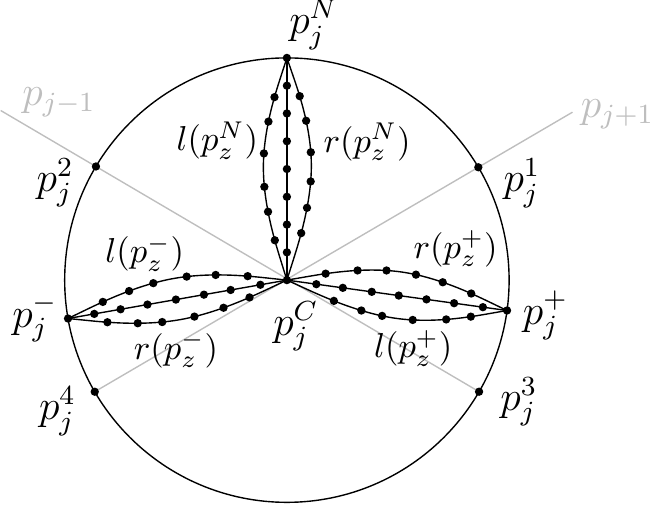}}
	\hfill
	\subfigure[\label{fig:universalpointset2-b}]{\includegraphics[width=0.49\textwidth]{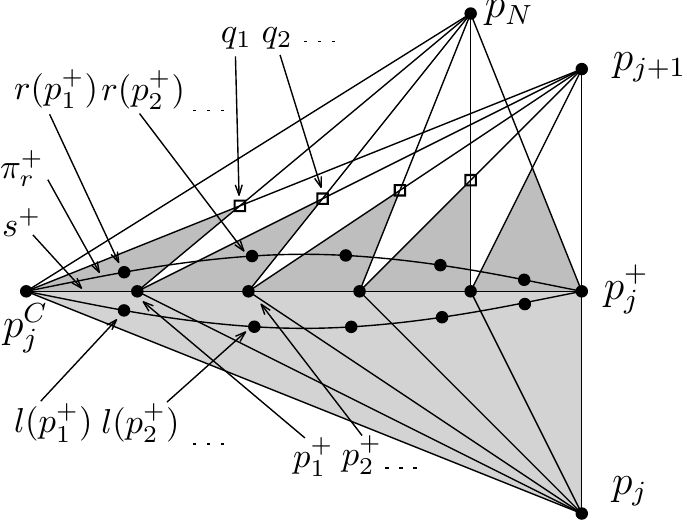}}
	\caption{Construction of petal points for $s^+$. Dark-gray triangles are used for petal points $r(p^{+}_z)$ while light-gray triangles for $l(p^{+}_z)$.}
	\label{fig:universalpointset2}
\end{figure}

We construct a point set $S^*$ with $O(n^{3/2})$ points from $S$ by adding 
\emph{petal points} to segments $s^+_j, s^N_j, s^-_j$ of the point sets $S_j$, 
for every $j$=$2,\dots,N-1$ (see 
Fig.~\ref{fig:universalpointset2-a}). For simplicity of notation, we skip the 
subscript $j$ whenever possible.
We denote by $p^\sigma_z$ the $z$-th point 
on segment $s^\sigma$, with $\sigma \in \{+,-,N\}$ and $z$=$1, \dots, \overline{n}$ (where 
$\overline{n}$=$\sqrt{n}$ or $\overline{n}$=$n$, depending on whether $p_j$ is sparse or dense), 
so that $p^\sigma_1$ is the point following $p^C$ along $s^\sigma$ and $p^\sigma_{\overline{n}} = p^\sigma_j$. 
For each point $p^\sigma_z$ we add two \emph{petal points} $l(p^{\sigma}_z)$ 
and $r(p^{\sigma}_z)$ to $S^*$. 

We first describe the procedure for $s^+$, see Fig.~\ref{fig:universalpointset2-b}. 
For each $z$=$1, \dots, \overline{n}$, consider the intersection point $q_z$ between segments $s(p^+_{z-1} p_{j+1})$ and 
$s(p^+_{z} p_N)$, where $p^+_{z-1} = p_j^C$ when $z=1$. 
By construction, all triangles $\triangle p^+_{z-1} p^+_{z} q_z$ have two corners on $s^+$, 
have the other corner in the same half-plane delimited by the line through $s^+$, 
and do not intersect each other except at common corners. Hence, there exists a 
convex arc $\pi^+_r$ passing through $p^C_j$ and $p^+_{\overline{n}}=p^+_j$, and 
intersecting the interior of every triangle. For each $z = 1, \dots, \overline{n}$, we 
place the petal point $r(p^{+}_z)$ on the arc of $\pi^+_r$ lying inside triangle 
$\triangle p^+_{z-1} p^+_{z} q_z$. For the other petal point
$l(p^{+}_z)$ we use the same procedure by 
considering triangles $\triangle p^+_{z-1} p^+_z p_j$ instead of $\triangle 
p^+_{z-1} p^+_z q_z$.
Symmetrically we place the petal points for $s^-$, using points 
$p_{j-1}$ and $p_1$ to place $l(p^{-}_z)$ and point $p_j$ to place $r(p^{-}_z)$, 
and for $s^N$, using points $p_{j-1}$ and $p_1$ to place $l(p^{N}_z)$ and points 
$p_{j+1}$ and $p_N$ to place $r(p^{N}_z)$.

Recall that we have $N = n+\sqrt{n}$ points $p_j$ on the outer half circle $\pi$ 
of $S$, and $N-2$ of them have their point set $S_j$. For each dense $p_j$ we 
added $6n$ petal points to $S^*$, while for every sparse $p_j$ we added $6\sqrt{n}$ 
petal points. Hence, the new point set $S^*$ has $(\sqrt{n}-1)(9n+1) + (n-1)(9\sqrt{n}+1)$=$O(n^{3/2})$ points.


\subsection{Modifying and Labeling the Graph}
\label{subsection:replacement}

We now aim at modifying $[G,{\cal H}]$ to obtain an inner-triangulated graph that 
can be embedded on the original point set $S$ ({\bf Part A} 
and {\bf Part B}); in Section~\ref{subsection:reverting} we describe how to 
exploit this embedding on $S$ to obtain an embedding of 
the original graph $[G,{\cal H}]$ on the extended point set $S^*$ ({\bf Part C}).
We describe the procedure just for a cycle-tree graph $[F,T]$ composed 
of a face $F$ of $G$ and of the tree $T$ inside it.

We first summarize the operations performed in the different Parts and then give more details in the following.

\begin{enumerate}
	\item {\bf Part A:}
	\begin{itemize}
		\item We delete some edges from $[F,T]$ connecting $F$ with $T$ 
		to identify ``tree components'', resulting in a new graph $[F,T'=T]$; note that 
		the set of edges connecting $T'$ to $F$ might be different from the set of edges 
		connecting $T$ to $F$. 
		\item We delete from $[F,T']$ the ``tree components'', to be defined later, and obtain 
		a new graph $[F,T'' \subseteq T']$ which has the property that it admits an augmentation to inner-triangulated without multiple edges. 
		\item We augment $[F,T'']$ to an inner-triangulated graph 
		$[F,T^{\Delta}=T'']$; again, instance $[F,T^{\Delta}]$ might differ from 
		$[F,T'']$ only on the set of edges connecting the two levels.
	\end{itemize} 
	\item We label $[F,T^{\Delta}]$ with the algorithm described in 
	Section~\ref{subsection:Labeling}.
	\item {\bf Part B:}
	\begin{itemize}
		\item We insert vertices in $[F,T^{\Delta}]$ representing the 
		previously removed tree components and give suitable labels to these vertices, 
		hence obtaining a new instance $[F,T^{\cal A} \supseteq T^{\Delta}]$. 
		By adding appropriate edges we keep the instance triangulated.
	\end{itemize}
	\item We embed $[F,T^{\cal A}]$ on point set $S$ 
	with the algorithm described in Section~\ref{subsection:Embedding}.
	\item {\bf Part C:}
	\begin{itemize}
		\item We obtain a planar embedding of $[F,T]$ on point set $S^*$ 
		by removing all the vertices and edges added during these steps and by suitably 
		adding back the removed edges and tree components.
	\end{itemize}
\end{enumerate}

\vspace{0.1 cm}
\noindent
{\bf Part A:}
We categorize each face $f$ of $[F,T]$ 
based on the number of vertices of $F$ and of $T$ that are incident to it. 
Since $T$ is a tree, $f$ has at least a vertex of $F$ and a vertex of $T$ 
incident to it. If $f$ contains exactly one vertex of $F$, then it is a \emph{petal face}.
If $f$ contains exactly one vertex of $T$, then it is a \emph{small face}.
Otherwise, it is a \emph{big face}.
Consider a big face $f$ and let $b_1,\dots,b_l$ be the occurrences of the vertices of $T$ in 
a clockwise order walk along the boundary of $f$. If either $b_1$ or $b_l$, say $b_1$, has more 
than one adjacent vertex in $F$ (namely one in $f$ and at least one not in $f$), 
then $f$ is \emph{protected} by $b_1$.
If $f$ is a big face with exactly two vertices incident to $F$ and is not 
protected by any vertex, then $f$ is a \emph{bad face}.

The next lemma gives sufficient conditions to triangulate $G$ without 
introducing multiple edges; we will later use this lemma to identify the ``tree 
components'' of $T$ whose removal allows for a triangulation.

\begin{figure}[t]
	\subfigure[\label{fig:triangulationedges-a}]{\includegraphics[width=0.19\textwidth]{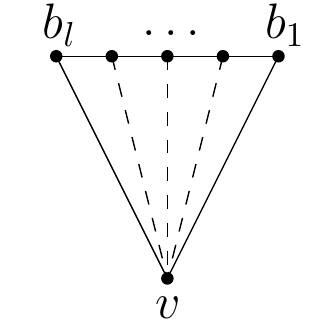}}
	\hfill
	\subfigure[\label{fig:triangulationedges-d}]{\includegraphics[width=0.19\textwidth]{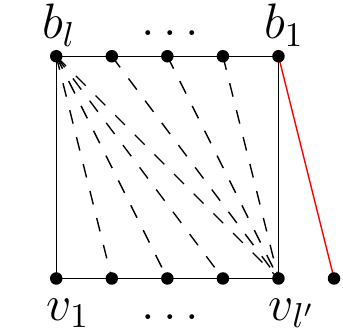}}
	\hfill
	\subfigure[\label{fig:triangulationedges-c}]{\includegraphics[width=0.19\textwidth]{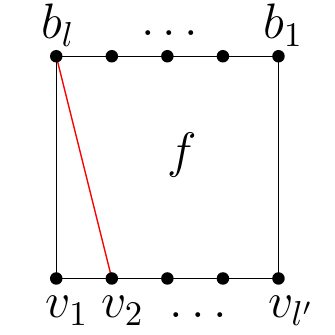}}
	\hfill
	\subfigure[\label{figap:bigfaces-a}]{\includegraphics[width=0.19\textwidth]{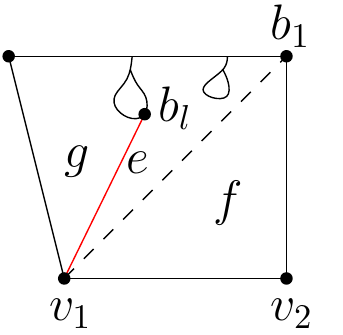}}
	\hfill
	\subfigure[\label{figap:bigfaces-b}]{\includegraphics[width=0.19\textwidth]{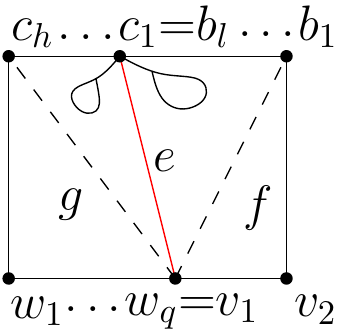}}
	\caption{\subref{fig:triangulationedges-a}--\subref{fig:triangulationedges-c}: Insertion of triangulation edges in \subref{fig:triangulationedges-a} 
		a petal face, \subref{fig:triangulationedges-d} a non-protected big face, and \subref{fig:triangulationedges-c} a big face protected by vertex $b_1$. \subref{figap:bigfaces-a}--\subref{figap:bigfaces-b} Illustration of the two cases for removing bad faces. Face $g$ is a petal face in \subref{figap:bigfaces-a} and a big face in \subref{figap:bigfaces-b}. 
		Dummy edges are dashed, while the removed edge $e$ is red.}
	\label{fig:triangulationedges}
\end{figure}

\begin{lemma}\label{lemma:triangulation} 
	Let $[F,T]$ be a biconnected simple cycle-tree graph, such that $(1)$ each 
	vertex of $F$ has degree at most four, and $(2)$ there exists no bad face in 
	$[F,T]$. It is possible to augment $[F,T]$ to an inner-triangulated simple cycle-tree 
	graph. 
\end{lemma}

\begin{proof}
	Let $f$ be any face of $[F,T]$. We describe how to triangulate $f$ without 
	creating multiple edges.
	
	Suppose $f$ is a petal face (see Fig.~\ref{fig:triangulationedges-a}); let $v,b_1,...,b_l$ (with $l>2$) be the vertices on 
	its boundary, where $v \in F$ and $b_i \in T$ for $1 \leq i \leq l$. We 
	triangulate $f$ by adding an edge $(v,b_i)$, for each $2 \leq i \leq l-1$. Since 
	$[F,T]$ is biconnected, there exists no multiple edge inside $f$. Also, since 
	condition (1) ensures that $v \in F$ has degree at most four, there is no petal 
	face incident to $v$ other than $f$, and thus no multiple edge is created 
	outside $f$. 
	
	Suppose $f$ is a small face; let $v_1,\dots,v_{l'},b$ (with $l>2$) be the vertices 
	on its boundary, where $v_i \in F$ for $1 \leq i \leq l'$ and $b \in T$. We 
	triangulate $f$ by adding an edge $(b,v_i)$, for each $2 \leq i \leq l'-1$. Note 
	that, before introducing these edges, vertices $v_2,\dots,v_{l'-1} \in F$ were 
	not connected to any vertex of $T$ (and in particular to $b$); thus, no multiple 
	edge is created.
	
	Suppose $f$ is a big face that is not a bad face; let 
	$v_1,...,v_{l'},b_1,...,b_l$ (with $l,l'>1$) be the vertices along the boundary 
	of $f$, where $v_1,...,v_{l'} \in F$ and $b_1,...,b_l \in T$.
	If $f$ is not protected by any vertex (see Fig.~\ref{fig:triangulationedges-c}), then $l'\geq 3$, as otherwise it would be a bad face. This implies that vertex $v_2 \in F$ 
	is not connected to any vertex of $T$. Hence, it is possible to add edge 
	$(b_l,v_2)$ without creating multiple edges. Face $f$ is hence split into a triangular 
	face $v_1,v_2,b_l$ and a big face that is protected by $b_l$, which we cover in the next case.
	Otherwise, $f$ is protected by a vertex. 
	If $f$ is protected by $b_1$ (see Fig.~\ref{fig:triangulationedges-d}), 
	then we triangulate $f$ by adding edges $(b_i,v_{l'})$, for $2 \leq i \leq l$ 
	and $(b_l,v_i)$, for $2\leq i \leq l'-1$.
	If $f$ is protected by $b_l$, 
	then we triangulate $f$ by adding edges $(b_i,v_1)$, for $1 \leq i \leq l-1$ and 
	$(b_1,v_i)$, for $2 \leq i \leq l'-1$.
	Note that, before introducing these edges, vertices $v_2,\dots,v_{l'-1} \in F$ 
	were not connected to any vertex of $T$ (and in particular to $b_1$ and $b_l$); 
	also, vertices $b_2,\dots,b_l$ (vertices $b_1,\dots,b_{l-1}$) were not connected 
	to $v_{l'}$ (resp. to $v_1$), $f$ was protected by $b_1$ (resp. $b_l$). 
	Thus, no multiple edge is created.
	
	Since by condition (2) there exists no bad face in $[F,T]$, all the possible 
	cases have been considered; this concludes the proof of the lemma.
	\qed
\end{proof}

We now describe a procedure to transform cycle-tree graph $[F,T]$ into 
another one $[F,T'']$ that is biconnected and satisfies the conditions of 
Lemma~\ref{lemma:triangulation}. We do this in two steps: first, we remove some 
edges connecting a vertex of $F$ and a vertex of $T$ to transform 
$[F,T]$ into a cycle-tree graph $[F,T'$=$T]$ that is not biconnected but that 
satisfies the two conditions; then, we remove the ``tree 
components'' of $T'$ that are not connected to vertices of $F$ in order to 
obtain a cycle-tree graph $[F,T'' \subseteq T']$ that is also biconnected.

To satisfy condition (1) of Lemma~\ref{lemma:triangulation}, 
we merge all the petal faces incident to the same vertex of $F$ into 
a single one by repeatedly removing an edge shared by two adjacent petal faces.
We refer to these removed edges as \emph{petal edges}, denoted by $E_P$.

To satisfy condition (2) of Lemma~\ref{lemma:triangulation}, we consider 
each bad face $f = v_1,v_2,b_1,\dots,b_l$, where $v_1,v_2 \in F$ and 
$b_1,\dots,b_l \in T$. Let $g$ be the face incident to $v_1$ sharing 
edge $e=(v_1,b_l)$ with $f$. We remove $e$, hence merging $f$ and $g$ 
into a single face $f'$, that we split again by adding dummy edges, 
based on the type of face $g$, in such a way that no new bad face is created.
Since $f$ is a bad face, it is not protected by $b_l$, and hence $g$ is not a small face.
If $g$ is a petal face, then $f'$ is still a big face with two vertices of $F$ 
incident to it, namely $v_1$ and $v_2$; see 
Fig.~\ref{figap:bigfaces-a}. We add edge $(v_1,b_1)$, 
splitting $f'$ into a petal face $v_1,b_1,\dots,b_l$ and a triangular face 
$v_1,v_2,b_1$.
If $g$ is a big face, then $f'$ is a big face; see 
Fig.~\ref{figap:bigfaces-b}. Let $w_1,\dots,w_q,c_1,\dots,c_h$ be the occurrences of vertices 
incident to $g$, where $w_1,\dots,w_q \in F$, with $w_q=v_1$, and $c_1,\dots,c_h 
\in T$, with $c_1=b_l$. We add two dummy edges $(v_1,c_h)$ and $(v_1,b_1)$, 
splitting $f'$ into a small face $w_1,\dots,w_{q},c_h$, a petal 
face $v_1,b_1,\dots,b_l=c_1,\dots,c_h$, and a triangular face $v_1,v_2,b_1$.
The edges removed in this step are \emph{big face edges}, 
denoted by $E_B$, and the added edges are \emph{triangulation edges}.

In order to make $[F,T']$ biconnected, note that $[F,T']$ 
consists of a biconnected component which contains $F$, called 
\emph{block-component}, and a set ${\cal T}_B$ of subtrees of $T'$, called 
\emph{tree components}, each sharing a cut-vertex with the block component.
We remove the tree components ${\cal T}_B$ from $[F,T']$ and obtain an instance 
$[F,T'' \subseteq T']$, that is actually the block component of $[F,T']$. Since the removal of ${\cal T}_B$ does not change the degree of the vertices of 
$F$ and does not create any bad face, $[F,T'']$ is indeed a biconnected 
instance that satisfies the two conditions of Lemma~\ref{lemma:triangulation}. 
Thus, we can augment it to an inner-triangulated instance $[F,T^{\Delta}]$, with 
$T^{\Delta}=T''$ by adding further \emph{triangulation edges}.
We state two important lemmas about $[F,T^{\Delta}]$.

\begin{lemma}\label{lemma:petaledgesandbigfaceedges}
	Let $e$=$(b,v)$ be an edge of $E_P \cup E_B$, where $b \in T$ and $v \in F$. Then, 
	either $e$ is a triangulation edge in $[F,T^{\Delta}]$ or $b$ 
	belongs to a tree component $T_c$ of ${\cal T}_B$ sharing a cut-vertex 
	$c$ with $[F,T'']$. In the latter case, $(v,c)$ is a 
	triangulation edge in $[F,T^{\Delta}]$.
\end{lemma}
\begin{proof}
	Suppose that $b \in T''$; we prove that $e$ is a triangulation edge in 
	$[F,T^{\Delta}]$.
	
	If $e \in E_P$, this directly descends from the fact that the algorithm to 
	triangulate a petal face $f$ described in Lemma~\ref{lemma:triangulation} adds a 
	triangulation edge between every vertex of $T$ incident to $f$, including $b$, 
	and the only vertex of $F$ incident to $f$, namely $v$.
	
	If $e \in E_B$, this depends again on the triangulation algorithm of 
	Lemma~\ref{lemma:triangulation} and on the addition of the one or two dummy 
	edges incident to $v$ that is performed when merging the two faces sharing edge 
	$e$. In fact, these dummy edges ensure that there exists a petal face in which 
	$v$ is the only vertex of $F$; then, the same argument as above applies to prove 
	that $v$ is connected to $b$ by a triangulation edge.
	
	Suppose that $b \notin T''$ and let $T_c$ be the tree component such that $b \in 
	T_c$; the fact that there exists a triangulation edge connecting $v$ to $c$ 
	follows from the same arguments as above, since in both cases $v$ is connected 
	by triangulation edges to all the vertices of $T$, including $c$, incident to 
	the same face it is incident to.
	\qed
\end{proof}

\begin{lemma}\label{lemma:singletarget}
	Let $T_c \in {\cal T}_B$ be a tree component such that there exists at least an 
	edge $(b,v) \in E_P \cup E_B$, with $b \in T_c$ and $v \in F$. Then, for 
	each edge in $E_P \cup E_B$ with an endvertex belonging to $T_c$, the other 
	endvertex is $v$.
\end{lemma}
\begin{proof}
	First suppose that all the edges in $E_P \cup E_B$ connecting a vertex of $T_c$ 
	to a vertex of $F$, including $e$, belong to $E_P$. Consider the two edges $e_1$ 
	and $e_2$ such that $e_1$ and $e_2$ connect $v$ to vertices of $T$, and all the 
	other edges that connect $v$ to vertices of $T$ lie between $e_1$ and $e_2$ in 
	the circular order of the edges around $v$ in $[F,T]$. Note that, all the edges 
	between $e_1$ and $e_2$ belong to $E_P$, while $e_1$ and $e_2$ do not, as one of 
	the two faces they are incident to is not a petal face. Let $f$ be the face both 
	$e_1$ and $e_2$ are incident to after the removal of all the edges between them. 
	Since all the vertices of $T_c$ are incident to $f$, and since $v$ is the 
	only vertex of $F$ incident to $f$, all the edges of $E_P$ 
	connecting a vertex of $T_c$ to a vertex of $F$ are incident to $v$.
	
	Suppose now that there exists at least an edge of $E_B$ connecting a vertex of 
	$T_c$ to a vertex of $F$. Hence, we can assume that 
	$e \in E_B$. This implies that $e$ is incident to a bad face $f$ and a face $g$ that can be 
	either a petal or a big face.
	
	If $g$ is a petal face, then let $e'=(v,b')$ be the other edge incident to $g$ 
	and to $v$. Since $g$ is a petal face, edge $e'$ belongs neither to $E_P$ nor to 
	$E_B$. Also, let $e''=(v,b'')$ be the dummy edge incident to $v$ added when 
	removing $e$ (the dashed edge in Fig.~\ref{figap:bigfaces-a}). Since, by construction, $e''$ is incident to a small face, it 
	belongs neither to $E_P$ nor to $E_B$, as well. Hence, both $e'$ and $e''$ are 
	edges of $[F,T']$ (and hence of $[F,T'']$) incident to $v$. This implies that 
	all the vertices of $T_c$ are incident to the unique face $g$ of $[F,T']$ to 
	which $e'$ and $e''$ are incident. Since $v$ is the only vertex of $F$ incident 
	to this face, all the edges of $E_P \cup E_B$ connecting a vertex of $T_c$ to a 
	vertex of $F$ are incident to $v$.
	
	If $g$ is a big face, then let $e'=(v,b')$ and $e''=(v,b'')$ be the two edges 
	incident to $v$ added when removing $e$ (the dashed edges in Fig.~\ref{figap:bigfaces-b}). Again, $e'$ and $e''$ belong to neither 
	$E_P$ nor $E_B$, since by construction they are both incident to small faces. 
	The statement follows by the same argument as above.
	\qed
\end{proof}

Performing the above operations for every cycle-tree graph $[F,T]$ yields an 
inner-triangulated $2$-outerplanar graph $[G,{\cal H}^{\Delta}]$, that is the 
outcome of {\bf Part A}. 

We then label $[G,{\cal H}^{\Delta}]$ with the algorithm described in 
Section~\ref{subsection:Labeling} and describe in the following how to extend this labeling to the tree components.
\vspace{0.1 cm}
\newline
\begin{wrapfigure}[10]{r}{26ex}
	\vspace{-3ex}
	\centering\includegraphics[width=0.24\textwidth]{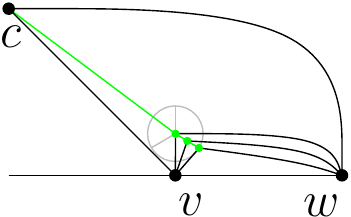}
	\caption{Inserting dummy vertices for a tree-component in face $(c,v,w)$ 
		with $v,w \in F$, $c \in T^\Delta$, $\ell(c) \leq \ell(v)$.} 
	\label{fig:tree-components-a-1}
\end{wrapfigure}
\noindent
{\bf Part B:}
We consider the tree components $T_c \in \mathcal{T}_B$ for each face $F$ of 
$G$; let $[F,T^{\Delta}]$ be the corresponding inner-triangulated cycle-tree graph. 
We label the vertices of $T_c$ and simultaneously augment $[F,T^{\Delta}]$ 
with dummy vertices and edges, so that $[F,T^{\Delta}]$ remains 
inner-triangulated (and 
hence can be embedded, by Lemma~\ref{theorem:2outerplanartriangulated}) 
and the vertices of $T_c$ can be later placed on the petal points 
of the points where dummy vertices are placed. 
The face of $[F,T'']$ to which $T_c$ belongs might 
have been split into several faces of $[F,T^{\Delta}]$ by 
triangulation edges. We assign $T_c$ to any of such faces $f$ that is incident 
to the root $c$ of $T_c$. Then, we label $T_c$ based on the type of $f$; we 
distinguish two cases. 

Suppose $f$ is a triangular face $(c,v,w)$ with $v,w  \in F$ and $c \in T^\Delta$, 
as in Fig.~\ref{fig:tree-components-a-1}; 
assume $\ell(v)<\ell(w)$. We create a path $P_c$ containing $|T_c|-1$ dummy vertices 
and append this path at $c$. Then, we connect every dummy vertex of $P_c$ with 
both $v$ and $w$. If $\ell(c) \leq \ell(v)$, 
then we label the vertices of $P_c$ with $\ell(P_c)=\ell(v)$. If $\ell(c) \geq \ell(w)$, 
then we label them with $\ell(P_c)=\ell(w)$. 

\begin{figure}[t]
	\subfigure[\label{fig:tree-components-b-1}]{\includegraphics[width=0.19\textwidth]{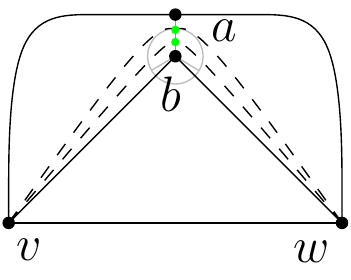}} 
	\hfill
	\subfigure[\label{fig:tree-components-b-2}]{\includegraphics[width=0.19\textwidth]{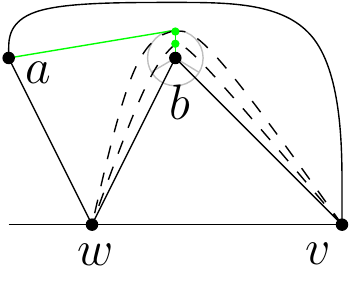}}
	\hfill
	\subfigure[\label{fig:tree-components-b-3}]{\includegraphics[width=0.19\textwidth]{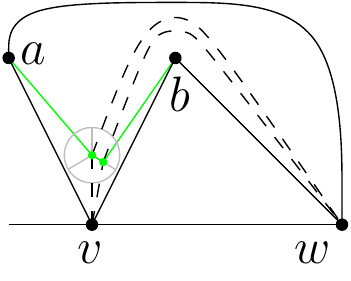}}
	\hfill
	\subfigure[\label{fig:placesfortreecomponents-b}]{\includegraphics[width=0.19\textwidth]{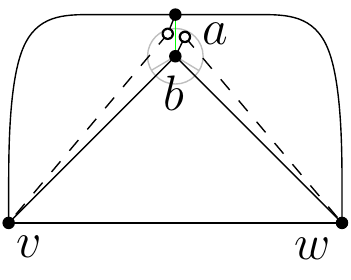}}
	\hfill
	\subfigure[\label{fig:placesfortreecomponents-a}]{\includegraphics[width=0.19\textwidth]{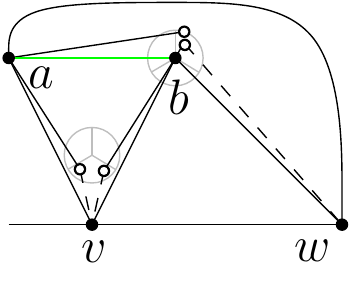}}
	\caption{\subref{fig:tree-components-b-1}--\subref{fig:tree-components-b-3} Inserting dummy vertices for a tree-component in a triangular face $(a,b,v)$ with $v \in F$ and $a,b \in T^\Delta$, when \subref{fig:tree-components-b-1} $\ell(a)=\ell(b)$, \subref{fig:tree-components-b-2} $\ell(a) \neq \ell(b)$ and $\ell(w) < \ell(v)$, and \subref{fig:tree-components-b-3} $\ell(a) \neq \ell(b)$ and $\ell(w) > \ell(v)$. \subref{fig:placesfortreecomponents-b}--\subref{fig:placesfortreecomponents-a} Moving dummy vertices to petal points if $\ell(a) = \ell(b)$ and if $\ell(a) \neq \ell(b)$, respectively.}
	\label{fig:tree-components-b}
\end{figure}

Suppose $f$ is a triangular face $(a,b,v)$ with $v \in F$ and $a,b \in 
T^\Delta$, refer to Fig.~\ref{fig:tree-components-b}; assume $\ell(a) \leq \ell(b)$. 
%
Replace edge $(a,b)$ with a path $P_c$ between $a$ and $b$ with $|T_c|-1$ 
internal dummy vertices, and connect each of them to $v$ and to 
$w$, where $w$ is the other vertex of $F$ adjacent to both $a$ and $b$.
For each dummy vertex $x$ of $P_c$, we assign $\ell(x) = \ell(a)$ if $\ell(v) \leq \ell(a)$; 
we assign $\ell(x) = \ell(b)$ if $\ell(v) \geq \ell(b)$; and we assign $\ell(x) = \ell(v)$ if 
$\ell(a) < \ell(v) < \ell(b)$.
The existence of edge $(a,b) \in T^\Delta$ implies that either $a$ is 
the parent of $b$ in $T^\Delta$ or vice versa. Suppose the former, the other 
case is analogous. Then, $v$ and $w$ 
are the extremal neighbors of $b$ in $F$, and thus either $\ell(v) \leq \ell(b) \leq \ell(w)$ 
or $\ell(w) \leq \ell(b) \leq \ell(v)$. 
Also, if $\ell(a) \neq \ell(b)$, then the label of $a$ does not lie strictly between 
those of $v$ and $w$. In fact, this can only happen if the label of $b$ 
strictly lies between those of $v$ and $w$, and $\ell(a)=\ell(b)$ (which happens only 
if $a$ is a non-fork vertex). Since $\ell(a) \leq \ell(b)$, by assumption, this 
implies that $\ell(a) \leq \ell(v),\ell(w)$.
The two observations before can be combined to conclude that, if $\ell(a)=\ell(b)$, 
then all the tree components lying inside faces $(a,b,v)$ and $(a,b,w)$ have the 
same label as $a$ and $b$ (Fig.~\ref{fig:tree-components-b-1}). Otherwise, 
either the tree components inside $(a,b,v)$ have label $\ell(b)$ and those inside 
$(a,b,w)$ have label $\ell(w)$ (Fig.~\ref{fig:tree-components-b-2}), or the tree 
components inside $(a,b,v)$ have label $\ell(v)$ and those inside $(a,b,w)$ have 
label $\ell(b)$ (Fig.~\ref{fig:tree-components-b-3}).

\noindent All added edges connecting a dummy vertex to $v$ and $w$ are 
again \emph{triangulation edges}.

We apply {\bf Part B} to every cycle-tree graph 
$[F,T^\Delta]$ of $[G,{\cal H}^\Delta]$, hence creating an 
inner-triangulated $2$-outerplanar graph $[G,{\cal H}^{\cal A}]$ where 
${\cal H}^{\cal A}$ is a forest. 
Since all the dummy vertices of $P_c$ are connected to two vertices $v,w \in F$, 
they become non-fork vertices.
Note that the labeling of the dummy vertices coincides with the 
one that would have been obtained by algorithm in Section~\ref{subsection:Labeling}, 
except for the case when $f$ is a triangular face $(a,b,v)$ with $v 
\in F$ and $a,b \in T^\Delta$, and $\ell(a) < \ell(v) < \ell(b)$. In this case, indeed, 
the algorithm would have assigned to $P_c$ label either 
$\ell(a)$ or $\ell(b)$, depending on whether $b$ is the parent of $a$ or vice versa. 
However, the fact that $\ell(a) < \ell(v) < \ell(b)$ holds in $[F,T^\Delta]$,
and the fact that $(a,b,v)$ is a triangular face of $[F,T^\Delta]$ imply that no vertex of 
$[F,T^\Delta]$ different from $v$ has been assigned the same label as $v$.
From these two observations we conclude that the restricted subgraph $H_i$ of $[G,{\cal H}^{\cal A}]$ for each $i$ is a tree with at most one vertex of degree larger than $2$, which has degree $3$. 
We thus apply Lemma~\ref{theorem:2outerplanartriangulated} to 
obtain a planar embedding $\Gamma^{\cal A}$ of $[G,{\cal H}^{\cal A}]$ on $S$.

\subsection{Transformation of the Embedding}
\label{subsection:reverting}

We remove the all the triangulation edges added in the construction, and then restore each tree component $T_c$, which is represented by path $P_c$. Since the vertices of $P_c$ are non-fork vertices and have the same label $i$, by construction, they are 
placed on the same segment $s \in \{s^+, s^N, s^-\}$ of $S_j$, where 
$p_j$ is the point vertex $v_i$ is placed on. 

We remove all the internal edges of $P_c$ and move each vertex $x$ of $P_c$ from 
the point $p$ of $s$ it lies on to one of the corresponding petal points, either 
$l(p)$ or $r(p)$, as follows. 
Let $v$ be a vertex of $G$ connected to a vertex of $T_c$ by an edge in $E_P 
\cup E_B$, if any; recall that, by Lemma~\ref{lemma:singletarget}, all the edges 
of $E_P \cup E_B$ connecting $T_c$ to $G$ are incident to $v$. 
If $\ell(x) < \ell(v)$, then move $x$ to $r(p)$; tree components connected to $w$ in 
Fig.~\ref{fig:placesfortreecomponents-b} and~\ref{fig:placesfortreecomponents-a}. If $\ell(x) > \ell(v)$, then move $x$ to $l(p)$; tree component connected to $v$ in 
Fig.~\ref{fig:placesfortreecomponents-a}.  
Otherwise, $\ell(x) = \ell(v)$; in this case $s \neq s^N$, by construction, and hence we have to 
distinguish the following two cases: If $s = s^+$, then move $x$ to $l(p)$, otherwise move $x$ to 
$r(p)$ (tree components attached to $a$ and $b$, respectively, and connected to $v$ in 
Fig.~\ref{fig:placesfortreecomponents-a}). 
If no vertex $v \in G$ is connected to $T_c$, then move $x$ to $r(p)$ if 
$\ell(c) < \ell(x)$ (tree component attached to $a$ in 
Fig.~\ref{fig:placesfortreecomponents-a}), and to $l(p)$ otherwise.

We prove that this operations maintain planarity. The internal edges of $T_c$ do not cross since the petal points, together with the point where $c$ lies, form a convex point set, on which it is possible to construct a planar embedding of every tree~\cite{BinucciGDEFKL10}. As for the edges connecting vertices of $T_c$ to $v$, by 
Lemma~\ref{lemma:petaledgesandbigfaceedges}, $v$ has visibility to the root $c$ 
of $T_c$, since $(v,c)$ is a triangulation edge; by 
Property~\ref{property:visibility}, this visibility from $v$ extends to all the segment $s$ where $P_c$ had been placed on; and by the construction of $S^*$, to all the corresponding petal points.
Hence, we only have to prove that the edges $(a,b)$ that had been subdivided into a path $P_c$ when merging tree component $T_c$ (green edges in Fig.~\ref{fig:placesfortreecomponents-b} and~\ref{fig:placesfortreecomponents-a}) can be reinserted without introducing any crossing. Namely, let $v$ and $w$ be the two vertices of $G$ that are connected to both $a$ and $b$. Recall that all the subdivision vertices of $(a,b)$ correspond to vertices of tree components belonging to faces $(a,b,v)$ and $(a,b,w)$.
If $\ell(a)=\ell(b)$ (see Fig.~\ref{fig:placesfortreecomponents-b}), then for each 
tree component $T_c$ belonging to face either $(a,b,v)$ or $(a,b,w)$, the 
vertices of $P_c$ lie on the segment $s^N$ corresponding to $\ell(a)=\ell(b)$, by 
construction, since they are non-fork vertices on the path between $a$ and $b$ 
and have label $\ell(a)=\ell(b)$. Also, both $a$ and $b$ lie on $s^N$, possibly at its extremal points. Since, by construction, all the 
tree components that are connected to $v$ (to $w$) through edges of $E_P \cup 
E_B$ are moved to petal points lying inside triangle $\triangle(a,b,v)$ (triangle 
$\triangle(a,b,w)$), and since no tree component stays on $s^N$, edge $(a,b)$ does 
not cross any edge.
If $\ell(a) \neq \ell(b)$, 
the fact that edge $(a,b)$ does not cross any edge 
again depends on the labels we assigned to the tree 
components belonging to faces $(a,b,v)$ and $(a,b,w)$.
Namely, assume that $\ell(a) < \ell(b)$ and that $a$ is the parent of $b$ (see 
Fig.~\ref{fig:placesfortreecomponents-a}), the other cases being analogous. As 
observed above, either the tree components belonging to $(a,b,v)$ have label 
$\ell(b)$ and those belonging to $(a,b,w)$ have label $\ell(w)$, or the tree 
components belonging to $(a,b,v)$ have label $\ell(v)$ and those belonging to 
$(a,b,w)$ have label either $\ell(b)$.
We prove the claim in the latter case (as in the figure), the other being 
analogous. Note that, for each tree component $T_c$ belonging to face $(a,b,w)$, 
all the vertices of $P_c$ lie on the segment $s^N$ corresponding to $\ell(b)$, by 
construction, since they are non-fork vertices on the path between $a$ and $b$ 
and have label $\ell(b)$. Hence, Property~\ref{property:visibility} ensures that 
they lie inside triangle $\triangle(a,b,w)$, which implies that the corresponding 
petal points lie inside $\triangle(a,b,w)$, as well. The fact that the tree 
components $T_c$ lying inside face $(a,b,v)$ are also placed on petal points 
lying inside triangle $\triangle(a,b,v)$ trivially follows from the fact that the 
vertices of $P_c$ have label $\ell(v)$.

To complete the transformation it remains to insert the edges of $E_P \cup E_B$ 
which were not inserted in the previous step. Since by Lemma~\ref{lemma:petaledgesandbigfaceedges} 
all of these edges were also triangulation edges, their insertion does not produce any crossing.

\begin{lemma}\label{theorem:2cycleouterplanar}
	There exists a universal point set of size $O(n^{3/2})$ for the class of 
	$n$-vertex $2$-outerplanar graphs $[G,{\cal H}]$ where ${\cal H}$ is a forest.
\end{lemma}

\section{General $2$-Outerplanar Graphs}\label{section:contraction}

In this section we extend the result of Lemma~\ref{theorem:2cycleouterplanar} to 
any arbitrary $2$-outerplanar graph $[G,{\cal H}]$. 

We first give a high-level description of the algorithm and then go into details. 
The main idea is to convert every graph $G_h \in {\cal H}$ lying in a face $F=F_h$ of $G$ into a tree $T_h$; embed 
the resulting graph on $S^*$; and finally revert the conversion from each $T_h$ 
to $G_h$. Each tree $T_h$ is created by 
substituting each biconnected block $B$ of $G_h$ by a star, 
which is centered at a dummy vertex and has a leaf for each vertex of $B$, where 
leaves shared by more stars are identified with each other. This results in a 
$2$-outerplanar graph whose inner level is a forest. 

The embedding of this graph on $S^*$ is performed similarly as in 
Lemma~\ref{theorem:2cycleouterplanar}, with some slight modifications to the 
labeling algorithm, especially for the vertices of $T_h$ corresponding to 
cut-vertices of $G_h$, and to the procedure for merging the tree components. 
These modifications allow us to ensure that the leaves of each star composing $T_h$, 
and hence the vertices of each block of $G_h$, lie on a portion of 
$S^*$ determining a convex point set, where they can thus be drawn without 
crossings~\cite{comgeo/Bose02,GritzmannMPP91}.

We now describe the arguments more in detail, starting by giving some definitions. 
We say that a cut-vertex of $G_h$ is a \emph{c-vertex}, and that the vertices and the edges of a block $B$ of $G_h$ are its \emph{block vertices}, denoted by $N_B$, and its \emph{block edges}, denoted by $E_{BL} \subseteq N_B \times N_B$, respectively. 
Now we transform graph $[F,G_h]$ into a cycle-tree graph $[F,T]$ as follows: For each block $B$ of $G_h$, we remove all its block edges $E_{BL}$ and insert a \emph{b-vertex} $b$ representing $B$; also, we insert edges $(b,b')$ for every vertex $b' \in N_B$. In other words, we replace each block $B$ with a star whose center is a new vertex $b$ and whose leaves are the vertices in $N_B$. This results in transforming $G_h$ into a tree $T$ obtained by attaching the stars through the identification of leaves corresponding to c-vertices. When performing the transformation, we start from the given planar embedding $\Gamma$ of $[G,{\cal H}]$, which naturally induces a planar embedding $\Gamma'$ of each resulting cycle-tree graph $[F,T]$. 

We apply the operations described in {\bf Part A} of Section~\ref{subsection:replacement} (delete petal and big-face edges, remove tree components, and triangulate) to make $[F,T]$ inner-triangulated, and then label it as in Section~\ref{subsection:Labeling}. 
We will then relabel some of the c-vertices and perform the merging of the tree components in a special way, slightly different from the one described in {\bf Part B}, so that the embedding of the resulting graph will satisfy some additional geometric properties that will allow us to restore the original blocks of $G_h$ when performing {\bf Part C}. 

Let $w_1,\dots,w_m$ be the vertices of $F$ in the order defined by the labeling, and let $r$ be the root of $T$; recall that, since the root is a fork vertex, it is independent of where the tree components, which become non-fork vertices, are merged. We give some additional definition.
For a b-vertex $b$ we define two particular vertices, called its \emph{opener} and the \emph{closer}, that will play a special role in the merging of the tree components incident to $b$. 
If $b \neq r$ and $b$ is not adjacent to $r$, then the opener of $b$ is the c-vertex $c$ that is the parent of $b$ in $T$.
If $b = r$ (see Fig.~\ref{fig:rootblocks-a}), then the opener of $b$ is the c-vertex $c$ adjacent to $b$, $w_1$, and $w_m$, such that $3$-cycle $(c,w_1,w_m)$ does not contain in its interior any c-vertex with the same property as $c$ in $\Gamma'$.
If $b$ is adjacent to $r$, then the opener of $b$ is $r$; note that, in this way we treat $r$ as a c-vertex even when it is not a cut-vertex of $G_h$.
For a b-vertex $b$ with opener $c$, the \emph{closer} of $b$ is the first block vertex following (the last preceding) $c$ in the rotation at $b$ in $\Gamma'$, if $\ell(c) < \ell(b)$ (if $\ell(c) \geq \ell(b)$); note that, the closer always exists since $b$ has at least two neighbors that are not incident to $F$.

Some blocks of $G_h$, and the corresponding b-vertices of $T$, have to be treated in a special way because of their relationship with the root $r$ of $T$. Let $c$ be the opener of a b-vertex $b$ such that $N_B \cup \{b\}$ contains $r$, where $B$ is the block of $G_h$ corresponding to $b$. We call \emph{root-blocks} the set of blocks lying in the interior of $3$-cycle $(c,w_1,w_m)$ in $\Gamma'$. If $c$ is a non-fork vertex, the presence of root-blocks might create problems in the algorithm we are going to describe later; hence, in this case, we change the embedding $\Gamma'$ slightly (cf. Figure~\ref{fig:rootblocks-a}) by rerouting edge $(w_m,b)$ so that root-blocks do not exist any longer. This change of embedding consists of swapping edges $(b,c)$ and $(w_m,b)$ in the rotation at $b$. Note that edge $(w_m,b)$ does not belong to $[F,G_h]$, which implies that embedding $\Gamma$ has not been changed. In order to maintain planarity, we have to remove all the edges connecting $w_1$ to root-blocks, as otherwise they would cross edge $(w_m,b)$; however, the fact that $(w_m,b)$ does not belong to $[F,G_h]$, together with a visibility property between $w_1$ and the root-blocks that we will prove in Lemma~\ref{lemma:correctorder}, will make it possible to add the removed edges at the end of the construction without introducing any crossing.

\begin{figure}[tb]
	\centering
	\subfigure[\label{fig:rootblocks-a}]{\includegraphics[width=0.37\textwidth]{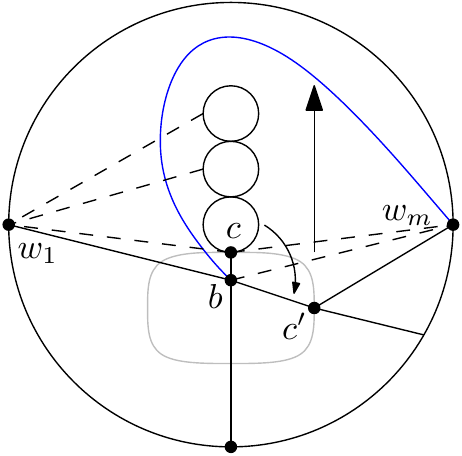}}
	\hfil
	\subfigure[\label{fig:rootblocks-b}]{\includegraphics[width=0.37\textwidth]{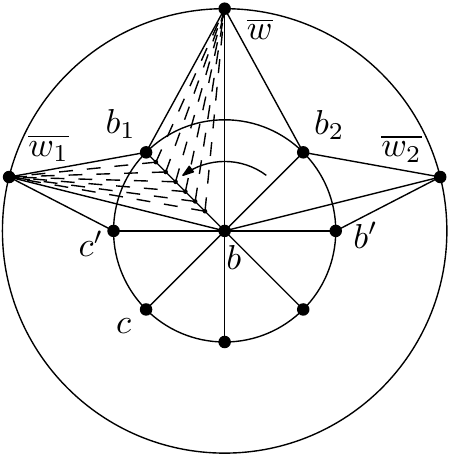}}
	\caption{\subref{fig:rootblocks-a} Rerouting edge $(w_m,b)$ to eliminate root-blocks when the opener $c$ of the block containing the root is a non-fork vertex. \subref{fig:rootblocks-b} Illustration for the rule ``choice of Faces''.} 
	\label{fig:rootblocks}
\end{figure}

We now describe the part of the algorithm that differs from the one described in Section~\ref{section:Forests}.

First, we change the labeling of each c-vertex $c$ that is a branch vertex of $T$. Namely, consider the two fork vertices $a$ and $d$ such that the subpath of $T$ between $a$ and $d$ contains $c$ and does not contain any other fork vertex, with $a$ being closer to the root than $d$. Let $v$ and $w$ be the two neighbors of $c$ in $F$; assume $\ell(w) < \ell(v)$. Note that, as described in \textbf{Part B} of Section~\ref{subsection:replacement}, we have either $\ell(w) < \ell(d) < \ell(v) \leq \ell(a)$ or $\ell(a) \leq \ell(w) < \ell(d) < \ell(v)$. In the first case, we relabel $c$ by setting $\ell(c) = \ell(v)$, otherwise we set $\ell(c) = \ell(w)$. Observe that this is analogous to considering $c$ as a tree component and applying for it the labeling algorithm in Section~\ref{subsection:replacement}. This observation allows us to state that the same arguments as in Lemma~\ref{lem:inducedsubtree} can be used to prove that the restricted subgraph $H_i$ of $G_h$, for each $i = 1, \dots, |G|$, maintains the same property even after the relabeling of $c$.

Then, we describe a procedure, that we call \textbf{Part B'} as it coincides with \textbf{Part B} of Section~\ref{subsection:replacement}, except for the choice of the face where the tree components are placed and of the edge they are merged to. This choice, that we describe later, is done in such a way that applying \textbf{Part C} of the embedding algorithm described in Lemma~\ref{theorem:2cycleouterplanar} yields an embedding $\Gamma^*$ of $[F,T]$ on $S^*$ that satisfies the following two properties, which will then allow us to redraw all the blocks of $G_h$: 
\begin{itemize}
	\item the block vertices of every block form a convex region and 
	\item the clockwise order in which the block vertices of every block appear along this convex region coincides with the clockwise order in which they appear along the outer face of the block in the drawing $\Gamma$ of $G$.
\end{itemize}

For ensuring the first item, the following important property derived 
from Property~\ref{property:visibility} is of particular help. 
Refer to Fig.~\ref{fig:convexityprperty}.

\begin{property}\label{property:convex}
	Let $\underline{j}$ and $\overline{j}$ be two integers such that $1 \leq \underline{j} < \overline{j} \leq N$. Then the points of 
	$\bigcup_{j = \underline{j}, \dots, \overline{j}} [s^-_{j} \cup \{p^C_{j}\} \cup s^+_{j}]$ determine a convex point set. This is also true if we replace $s^-_{\underline{j}}$ by $s^N_{\underline{j}}$ and $s^+_{\overline{j}}$ by $s^N_{\overline{j}}$.
\end{property}
\begin{proof}
	First observe that the center points $p_j^C$ of all the point sets between $\underline{j}$ and $\overline{j}$, that is, $\bigcup_{j = \underline{j}, \dots, \overline{j}} [\{p^C_{j}\}]$ are in convex position by construction.
	
	\begin{figure}[tb]
		\begin{center}
			\includegraphics[width=0.6\textwidth]{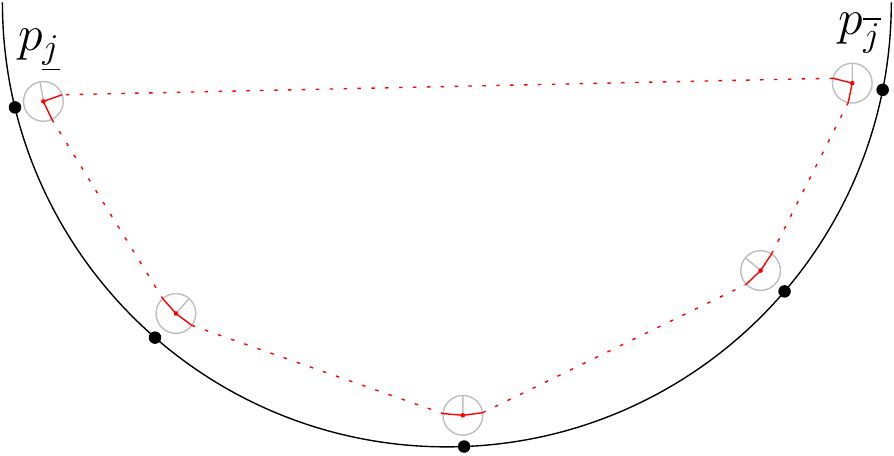}
		\end{center}
		\caption{Illustration for Property~\ref{property:convex}.} 
		\label{fig:convexityprperty}
	\end{figure}
	
	Then, for each $j = \underline{j}, \dots, \overline{j}-1$, segments $s_j^+$ and $s_{j+1}^-$ lie below the segment $s(p^C_{j},p^C_{j+1})$, due to the fact that points $p_j^+$ and $p_{j+1}^-$ lie below points $p_j^1$ on $\pi_j$ and $p_{j+1}^2$ on $\pi_{j+1}$, respectively; see Fig.~\ref{fig:universalpointset}. This implies that the internal angles at $p_j^+$ and $p_{j+1}^-$ are smaller than $180^\circ$. As for the internal angle at each center point $p^C_{j}$, this is still smaller than $180^\circ$ due to the fact that $p_j^+$ and $p_{j}^-$ lie above points $p_j^3$ and $p_j^4$ on $\pi_j$, respectively, which lie on a diameter of $\pi_j$.
	
	The fact that segments either $s^-_{\underline{j}}$ or $s^N_{\underline{j}}$, and either $s^+_{\overline{j}}$ or $s^N_{\overline{j}}$ do not destroy the convexity of the point set again descends from the fact that the internal angles at $p^C_{\underline{j}}$ and at $p^C_{\overline{j}}$ are always smaller than $180^\circ$.
\end{proof}	

The second item can be mostly ensured by choosing an appropriate face for the tree components. In fact, as already noted in Section~\ref{section:Forests}, the triangulation step performed after the removal of tree components splits the face where each tree component used to lie into several faces; while in Section~\ref{section:Forests} the choice among these faces was arbitrary, in this case we have to make a suitable choice, which will be based on the opener and the closer of the block the tree component belongs to. 

{\bf Rule ``choice of Faces'':}

Let $b$ be a b-vertex of a block $B$, and let $c$ and $c'$ be the opener and the closer of $b$, respectively. Also, let $b'$ be the last counterclockwise neighbor of $b$ different from $c$ such that $b' \in N_B$ and $\ell(b')=\ell(b)$ (possibly, $b'=c'$).

Consider any two neighbors $b_1$ and $b_2$ of $b$ such that $b_1,b_2 \in N_B$ and there exists no vertex $b_3 \in N_B$ of $b$ between $b_1$ and $b_2$ in the rotation at $b$. Since $[F,T]$ is inner-triangulated, there exists a vertex $\overline{w} \in F$ that is adjacent to both $b_1$ and $b_2$; also, there exists edge $(b,\overline{w})$, which is a triangulation edge. Hence, each tree component $T_{1,2}$ that used to lie between $b_1$ and $b_2$ has to be placed either inside face $(b,b_1,\overline{w})$ or inside $(b,b_2,\overline{w})$ in order to maintain the embedding of the graph before the triangulation. Finally, let $\overline{w}_1$ and $\overline{w}_2$ be the two vertices of $F$ preceding $b_1$ and following $b_2$ in the rotation at $b$, respectively.

If both $b_1$ and $b_2$ are between $b'$ and $c'$ in the rotation at $b$, then place $T_{1,2}$ inside face $(b,b_2,\overline{w})$ and merge it to edge $(b,b_2)$, that is, subdivide this edge with $|T_{1,2}|$ dummy edges, each connected to $\overline{w}$ and to $\overline{w}_2$; otherwise, place $T_{1,2}$ inside face $(b,b_1,\overline{w})$ and merge it to edge $(b,b_1)$, connecting the subdivision edges to $\overline{w}$ and to $\overline{w}_1$; see Fig.~\ref{fig:rootblocks-b}.

Let $[F,T^*]$ be the cycle-tree graph obtained after all the tree components have been merged. In the following lemma we prove that $[F,T^*]$ admits an embedding on $S^*$ satisfying the required geometric properties.

\begin{lemma}\label{lemma:correctorder}
	There exists an embedding $\Gamma^*$ of $[F,T^*]$ on $S^*$ in which, for each b-vertex $b$ corresponding to a block $B$ of $G_h$, the vertices of $N_B$ are in convex position and appear along this convex region in the same clockwise order as they appear along the outer face of $B$ in the given planar drawing $\Gamma$ of $G$.
\end{lemma}

\begin{proof}
	First, construct a straight-line planar embedding $\Gamma''$ of $[F,T^*]$ on $S^*$ by applying Lemma~\ref{theorem:2cycleouterplanar}. 
	
	We will now consider each block $B$ represented by a b-vertex $b$ in $T^*$ and analyze where the vertices $N_B$ are placed in $\Gamma''$ due to \textbf{Part C} of Lemma~\ref{theorem:2cycleouterplanar} and to the rule ``choice of faces'' described in \textbf{Part B'}, proving that the vertices in $N_B$ either already satisfy the required properties or can do so by performing some local changes to $\Gamma''$.
	
	The block vertices $N_B$ consist of the fork vertices $N_{f}$, of the non-fork vertices $N_{tc}$ obtained by merging tree components, and of the other non-fork vertices $N_{nf}$, which are also non-fork vertices of $[F,T]$. Note that sets $N_{f}$, $N_{nf}$, and $N_{tc}$ are disjoint, if we consider the root of a tree component not in $N_{tc}$. 
	
	We start with removing $b$ and its incident edges. Note that, in the local changes we possibly perform, the position of $b$ might be reused by another vertex. As orientation help we sometimes keep $b$ on its point, in particular in illustrations, until all its block vertices have been considered. 
	
	First suppose that $B$ belongs to the root-blocks. Recall that the c-vertex $c^*$ separating the root-blocks from the block containing the root $r$ is a fork vertex, since in the case it was a non-fork vertex we rerouted edge $(w_m,b)$, hence eliminating the root-blocks. Thus, all the vertices of the root-blocks have the same label as $c^*$ and are placed on the $s^N_j$ segment of the point set $S_j$ where $c^*$ is placed. Since each vertex $x$ of $B$ in $N_{tc}$ is moved to a petal point of $s^N_j$ by the algorithm described in {\bf Part C}, and since the petal points of the same segment are in convex position, by construction of $S^*$, the vertices of $B$ satisfy the required properties.
	
	Assume now that $B$ does not belong to the root-blocks. We distinguish two cases, based on whether $b$ is a fork vertex or not. Let $c$ and $c'$ be the opener and the closer of $b$, respectively, and assume $\ell(c) \geq \ell(b)$ (the other case is symmetric). Refer to Fig.~\ref{fig:nonforkblockcentervertex}. Let $j$ and $k$ be the indexes such that $c$ is placed on point set $S_j$ and $c'$ is placed on point set $S_k$.
	
	\begin{figure}[tbh]
		\begin{center}
			\subfigure[\label{fig:nonforkblockcentervertex-a}]{\includegraphics[
				width=0.24\textwidth]{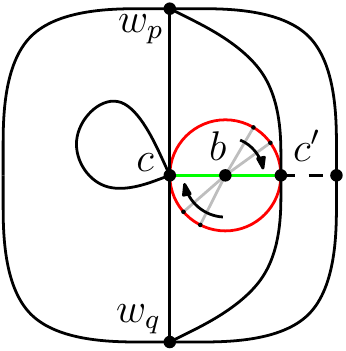}}
			\hfill
			\subfigure[\label{fig:nonforkblockcentervertex-b}]{\includegraphics[
				width=0.24\textwidth]{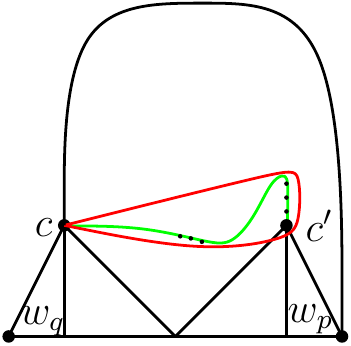}}
			\hfill
			\subfigure[\label{fig:nonforkblockcentervertex-d}]{\includegraphics[
				width=0.24\textwidth]{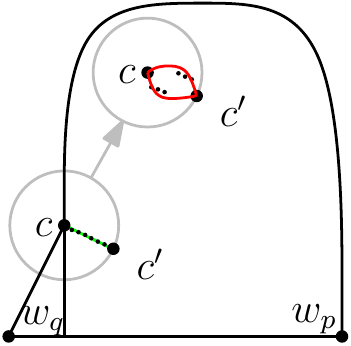}}
			\hfill
			\subfigure[\label{fig:nonforkblockcentervertex-c}]{\includegraphics[
				width=0.24\textwidth]{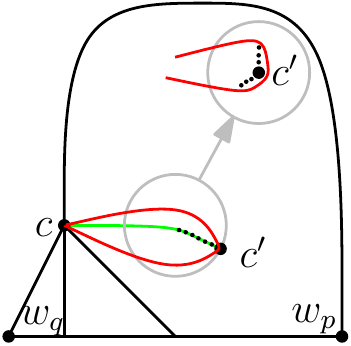}}
			\caption{Illustration when $b$ is non-fork vertex. The red circle indicates the block, 
				tiny black vertices are from tree-components, green edges are tree-edges and the dashed edge is 
				present of $c'$ is fork vertex, otherwise $c'$ is non-fork vertex. 
				The block illustrated in~\subref{fig:nonforkblockcentervertex-a} is placed as in~\subref{fig:nonforkblockcentervertex-b}, 
				if $c'$ is fork vertex. \subref{fig:nonforkblockcentervertex-d} illustrates the case $\ell(c) \neq \ell(c')$ 
				when a promotion of $c'$ is not necessary, while in \subref{fig:nonforkblockcentervertex-c} a promotion is necessary.}
			\label{fig:nonforkblockcentervertex}
		\end{center}
	\end{figure}
	%
	
	
	{\bf Suppose $b$ is a non-fork vertex}, and let $w_p,w_q$ (with $p < q$) be the neighbors of $b$ in $F$. Refer to Fig.~\ref{fig:nonforkblockcentervertex}.
	
	First note that, in this case, $c'$ is the only vertex of $N_B \setminus \{c\}$ belonging to $N_{f} \cup N_{nf}$, that is, all the vertices in $N_B$ different from $c$ and $c'$ belong to some tree components. Also, we have $\ell(c) \geq \ell(x) \geq \ell(c')$ for all $x \in N_B$. See Fig.~\ref{fig:nonforkblockcentervertex-a}.
	
	If $c' \in N_{f}$, then $c'$ is placed on the center point $p_k^C$ of $S_k$, as in Fig.~\ref{fig:nonforkblockcentervertex-b}. We have that the vertices of $N_B$ that have been merged to edge $(b,c')$ are placed on the $s^N_k$ segment of $S_k$, since the algorithm described in {\bf Part C} moved the vertices adjacent to $w_p$ inside triangle $(c,c',w_p)$; also, the vertices of $N_B$ that have been merged to edge $(b,c)$ are placed on the $s^-_l$ segment of a point set $S_l$ such that $k < l \leq j$, since the vertices adjacent to $w_q$ were moved inside triangle $(c,c',w_q)$. Hence, Property~\ref{property:convex} ensures that the vertices of $N_B$ are in convex position. The fact that they appear in the correct order along this convex region depends on the fact that the vertices merged to $(b,c')$, as well as those merged to $(b,c)$, are consecutive along the boundary of $B$.
	
	If $c' \in N_{nf}$, then $c'$ is placed on the $s^-_k$ segment of $S_k$. If $j=k$, as in Fig.~\ref{fig:nonforkblockcentervertex-d}, then $c$ is either on $s^-_k$ or on $p^C_k$; in both cases, the vertices in $N_B$ are on the same segment, and the proof that they satisfy the required properties, after they have been moved to petal points, is the same as for the case of the root-blocks. If $j > k$, as in Fig.~\ref{fig:nonforkblockcentervertex-c}, which can only happen if $c$ is a fork vertex, then all the points of $N_B$, except for $c$, lie on $s^-_k$, while $c$ lies on $p^C_j$. This implies that the region defined by the points of $N_B$ is not convex. We thus need to perform a local change in the placement of these vertices, that we call a \emph{promotion} of $c'$ at $S_k$. This operation places $c'$ on $p^C_k$, and places on $s^N_k$ the vertices of $N_B$ that were merged to $(b,c')$, and on $s^+_k$ the vertices of $N_B$ that were merged to $(b,c)$. Intuitively, this corresponds to ``promoting" $c'$ to become a fork vertex. Note that, no vertex lies on $p^C_k$ before the promotion of $c'$, since there is no fork vertex between $c$ and $c'$ in $T^*$, and this implies that no vertex lies on $s^N_k$ and $s^+_k$, as well. By Property~\ref{property:convex}, the vertices of $N_B$ are now in convex position and in the correct order, as in the case in which $c'$ is a fork vertex.
	
	{\bf Suppose $b$ is a fork vertex}, and let $w_p,w_q$ (with $p < q$) be the two extremal neighbors of $b$ in $F$. Refer to Fig.~\ref{fig:CaseForkFork}.
	
	\begin{figure}[tb]
		\begin{center}
			\includegraphics[width=1.0\textwidth]{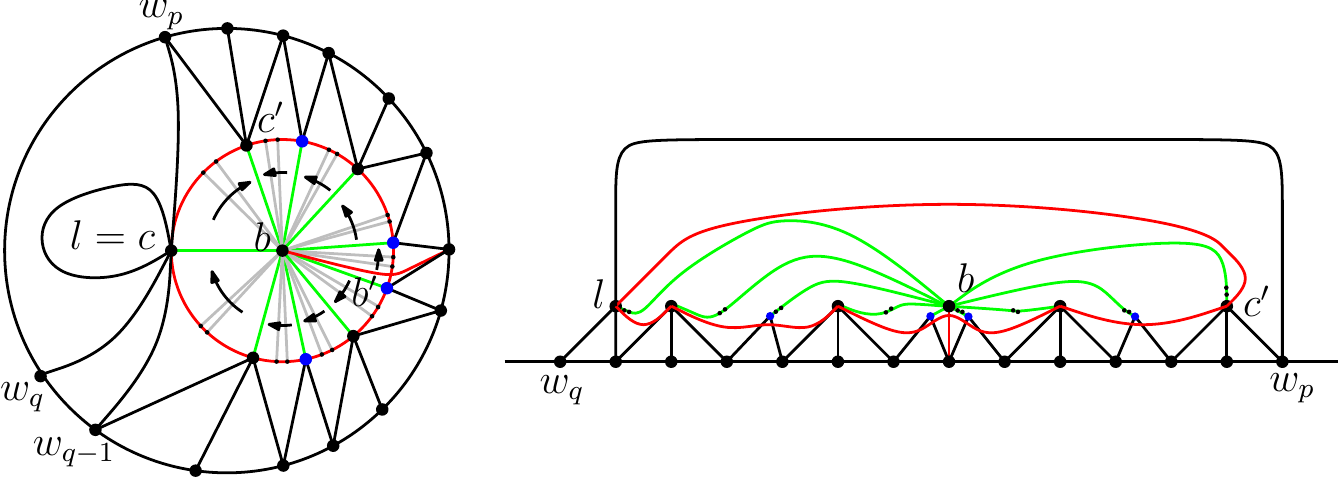}
		\end{center}
		\caption{Illustration when $c$ is fork vertex. The red circle indicates the block, 
			black(blue) vertices on this circle are fork(non-fork)-vertices. Arrows indicate to which 
			edge tree-components are assigned. The right drawing simulates the placement on the point set.} 
		\label{fig:CaseForkFork}
	\end{figure}
	
	Let $a$ be the ancestor of $b$ in $T$ such that $a$ is a fork vertex and there exists no fork vertex in the path of $T^*$ between $a$ and $b$. Note that, $a$ might either coincide with $c$ or it might be the b-vertex or the opener of an ancestor block $\overline{B}$ of $B$. In any case, vertex $a$ always exists, as the root $r$ is a fork vertex, except for the case in which $b$ itself is the root. This special case $b=r$ will be considered at the end of the proof. Also note that $a$ is adjacent to both $w_p$ and $w_q$, and we have $\ell(a) \geq \ell(x) \geq \ell(c')$ for all $x \in N_B$.
	
	We claim that $\ell(c) \geq \ell(x) \geq \ell(c')$ for all $x \in N_B$. Namely, if $c$ is a fork vertex, then $c=a$ and the claim trivially follows; while if $c$ is a non-fork vertex, then it is a branch vertex (since it has at least a fork vertex descendant, namely $b$), and hence it has been relabeled so that $\ell(c)=\ell(w_q)$.
	
	We then claim that, for each point set $S_l$ with $k < l \leq j$, there exists no vertex of $N_B$ lying on segment $s^N_l$. Namely, the embedding algorithm places a vertex $z$ on the $s^N_l$ segment only if $z$ is a branch vertex of $T$; however, this implies that there exists at least a child block of $B$ attached to $z$, and hence $z$ is the opener of this block. Thus, $z$ has been relabeled and does not lie on $s^N_l$.
	
	Finally, we consider the placement of $c'$ and of the tree components merged to edge $(b,c')$. If $c'$ is a fork vertex, then $c'$ lies on $p^C_k$, the vertices of $N_{tc}$ adjacent to $w_p$ are on $s^N_k$, and the other vertices of $N_{tc}$ are either on $s^+_{k}$ or on a segment $s^-_{k'}$, for some $k' > k$, by the algorithm described in \textbf{Part C}. If $c'$ is a non-fork vertex, then $c'$ lies on $s^-_{k}$, together with all the vertices of $N_{tc}$ that have been merged to $(b,c')$. We hence perform a promotion of $c'$ at $S_k$, moving $c'$ to $p^C_k$, the vertices of $N_{tc}$ adjacent to $w_p$ to $s^N_k$, and the other vertices of $N_{tc}$ to $s^+_{k}$. As in the previous case, there was no vertex of $N_B$ placed on $p^C_k$ before promoting $c'$; in this case, however, we have to consider the possibility that vertex $b$ was placed on $p^C_k$. Since $b$ has been removed, $p^C_k$ is again free, but a vertex of $N_B$ might still lie on $s^+_{k}$, namely $b'$. This does not affect the possibility of performing the promotion of $c'$, as we have only to ensure that $b'$ is moved on $s^+_{k}$ far enough from $p^C_k$ so that the other vertices of $N_B$ that are moved to that segment can fit. This is always possible since $s^+_{k}$ contains $\overline{n}$ points, where either $\overline{n} = \sqrt{n}$ or $\overline{n}=n$, and there exist at most $\overline{n}$ vertices in total on $S_k$.
	
	The two claims above, together with the discussion about $c'$, make it possible to apply Property~\ref{property:convex} to prove that the vertices of $N_B$ are in convex position. 
	
	In the following we prove that they appear along this convex region in the correct order. First note that the vertices in $N_{f}\cup N_{nf}$ are in the correct order, by construction. As for the vertices in $N_{tc}$, the algorithm in \textbf{Part C} places each set of vertices belonging to the same tree component $T_b$ between the two vertices of $N_{f}\cup N_{nf}$ incident to the face to which the vertices of $T_b$ have been assigned by the rule ``choice of faces" in \textbf{Part B'}. The only exception concerns the vertices merged to $(b,c')$ that are adjacent to $w_p$, as these vertices are on $s^N_k$; however, this is still consistent with the order in which the vertices of $N_B$ appear along the boundary of $B$.
	
	This concludes the proof of the lemma.
	\qed
\end{proof}

By Lemma~\ref{lemma:correctorder} the block vertices of every block are in convex position. Since every convex point of size $n$ set is universal for $n$-vertex outerplanar graphs~\cite{GritzmannMPP91,comgeo/Bose02}, we can now insert all block edges $E_{BL}$ in $\Gamma''$ without introducing any crossing. The resulting drawing is a planar embedding of $[F,G_h]$ on $S^*$, which proves the following.

\begin{lemma}\label{th:2-outerplanar-dense-sparse}
	Any $2$-outerplanar graph admits a planar straight-line embedding on a point set of size $O(n^{3/2})$. 
\end{lemma}

Using the technique from~\cite{AngeliniBKMRS11} we can reduce the size 
of $S^*$ to $O(n (\frac{\log n}{\log\log n})^2)$, but an even better bound can be obtained 
by using the super-pattern sequence $\xi$ from~\cite{DBLP:journals/jgaa/BannisterCDE14}, which allows us
to reduce the size of $S^*$ to $O(n \log n)$ points. Namely, this sequence $\xi$ 
of integers $\xi_j$, with $\sum_{j=1,\dots,n} \xi_j = O(n \log n)$, 
is a majorization of every sequence of integers that sum up to $n$. 
We hence assign the size of each point set $S_j$ based on this sequence,  
instead of using only dense or sparse point sets. We formalize this in the following 
theorem, which states the final result of the paper. 

\begin{theorem}\label{theorem:final}
	There exists a universal point set of size $O(n \log n)$ for the class 
	of $n$-vertex $2$-outerplanar graphs.
\end{theorem}

\begin{proof}
	Bannister et al.~\cite{DBLP:journals/jgaa/BannisterCDE14} proved that there exists a sequence $\xi$ of integers $\xi_j$, with $\sum_{j=1,\dots,n} \xi_j = O(n \log n)$, that satisfies the following property. For each finite sequence $\alpha_1,\dots,\alpha_k$ of integers such that $\sum_{i=1,\dots,k} \alpha_i = n$, there exists a subsequence $\beta_1,\dots,\beta_k$ of the first $k$ elements of $\xi$ such that, for each $i=1,\dots,k$, we have $\alpha_i \leq \beta_i$. 
	
	Bannister et al.~\cite{DBLP:journals/jgaa/BannisterCDE14} used this sequence to construct a universal point set of size a $O(n \log n)$ for simply-nested graphs~\cite{AngeliniBKMRS11}. We use the same technique to construct our universal point set $S^*$. Namely, for each $j = 1, \dots, n$, we place $\xi_j$ points on each of segments $s_j^-$, $s_j^+$, and $s_j^N$ of $S_j$, which hence results in a point set of total size $O(n \log n)$. Then, when each vertex $v_i \in G$ has to be placed on a point of the outer half-circle $\pi$ according to its weight $\omega(v_i)$, we place it on the first free point $p_j$ such that $\omega(v_i) \leq \xi_j$. Since the sum of the weights of the vertices of $G$ is equal to $n$, by the property of sequence $\xi$ we have that all the vertices of $G$ can be placed on $S^*$. This concludes the proof of the theorem.
\end{proof}
 
\section{Conclusions}\label{section:conclusions}

We provided a universal point set of size $O(n \log n )$ for $2$-outerplanar graphs. A natural question is whether our techniques can be extended to other meaningful classes of planar graphs, such as $3$-outerplanar graphs. We also find interesting the question about the required area of universal point sets. In fact, while the integer grid is a universal point set for planar graphs with $O(n^2)$ points and $O(n^2)$ area, all the known point sets of smaller size, even for subclasses of planar graphs, require a larger area. We thus ask whether universal point sets of subquadratic size require polynomial or exponential area. 

\bibliographystyle{abbrv}
%

\bibliography{bibliography}

\end{document}